\documentclass[11pt]{article}

\usepackage{amssymb,amsmath}
\usepackage[english]{babel}

\newtheorem{DE}{Definition}[section]

\newcommand {\sm} {\setminus}

\usepackage{latexsym} 
\usepackage{theorem} 
\newcommand{\qed}{\relax\ifmmode\hskip2em\Box\else\unskip\nobreak\hfill$\Box$\fi}

\newtheorem{theorem}[DE]{Theorem}
\newtheorem{lemma}[DE]{Lemma}

\newtheorem{corollary}[DE]{Corollary}
{\theoremstyle{break}\theorembodyfont{\rmfamily}}
{\theoremstyle{break}\theorembodyfont{\rmfamily}}

\newcounter{claim}
\newenvironment{proof}[1][]%
	{\noindent {\setcounter{claim}{0}\sc proof --- }{#1}{}}{\qed\vspace{2ex}}
	{\refstepcounter{claim}\vspace{1ex}\noindent {(\it\arabic{claim}) {#1}{}}\it}{\vspace{1ex}}
	{\noindent {}{#1}{}}{ This proves~(\arabic{claim}).\vspace{1ex}}

\begin{document}

\title{Graphs that do not contain a cycle with a node that has at
least two neighbors on it}

\author{Pierre Aboulker\thanks{Universit\'e Paris 7, LIAFA, Case 7014,
    75205 Paris Cedex 13, France.  E-mail:
    pierre.aboulker@liafa.jussieu.fr.} ,
Marko Radovanovi\'c\thanks{Faculty of Computer Science (RAF), Union
    University, Knez Mihailova 6/VI, 11000 Belgrade, Serbia. E-mail:
mradovanovic@raf.edu.rs.  Supported by Serbian Ministry of Education and Science project 174033} ,
\\Nicolas Trotignon\thanks{CNRS, LIP, ENS Lyon, INRIA, Universit\'e de Lyon
    (France),\ E-mail: nicolas.trotignon@ens-lyon.fr.}~
and Kristina
  Vu\v{s}kovi\'c\thanks{School of Computing, University of Leeds,
    Leeds LS2 9JT, UK and 
Faculty of Computer Science (RAF), Union
    University, Knez Mihajlova 6/VI, 11000 Belgrade, Serbia. E-mail:
    k.vuskovic@leeds.ac.uk.  Partially supported by 
 EPSRC grant EP/H021426/1
and Serbian Ministry
    of Education and Science projects 174033 and III44006.
\newline The first and third authors are partially supported by \emph{Agence Nationale de la Recherche} under reference
    \textsc{anr 10 jcjc 0204 01}
 \newline The four authors are also supported by PHC Pavle Savi\'c grant, 
jointly awarded by EGIDE, 
an agency of the French Minist\`ere des Affaires \'etrang\`eres et 
europ\'eennes, and Serbian 
Ministry of Education and Science}}

\date{June 19, 2012}

\maketitle

\begin{abstract}
  We recall several known results about minimally 2-connected graphs,
  and show that they all follow from a decomposition theorem.
  Starting from an analogy with critically 2-connected graphs, we give
  structural characterizations of the classes of graphs that do not
  contain as a subgraph and as an induced subgraph, a cycle with a
  node that has at least two neighbors on the cycle.  From these
  characterizations we get polynomial time recognition algorithms for
  these classes, as well as polynomial time algorithms for
  vertex-coloring and edge-coloring. 
\end{abstract}

\section{Introduction}

In this paper all graphs are finite, simple and undirected.  A
\emph{propeller} $(C,x)$ is a graph that consists of a chordless cycle
$C$, called the {\em rim}, and a node $x$, called the {\em center},
that has at least two neighbors on $C$.  The aim of this work is to
investigate the structure of graphs defined by excluding propellers as
subgraphs and as induced subgraphs.

In Section \ref{sec:mc2cg} we motivate the study of these two classes
of graphs by revisiting several theorems concerning classes of graphs
defined by constraints on connectivity, such as minimally and
critically 2-connected graphs.

Our second motivation for the study of propeller-free graphs is our
interest in wheel-free graphs.  A {\em wheel} is a propeller whose rim
has length at least~4 and whose center has at least 3 neighbors on the
rim.  We say that a graph $G$ {\em contains} a graph $F$ if $F$ is
isomorphic to a subgraph of $G$, and $G$ {\em contains $F$ as an
  induced subgraph} if $F$ is isomorphic to an induced subgraph of
$G$. We say that $G$ is {\em $F$-free} if $G$ does not contain $F$ as
an induced subgraph, and for a family of graphs $\mathcal F$, $G$ is
$\mathcal F$-free if it is $F$-free for every $F \in \mathcal F$.
Clearly, propeller-free graphs form a subclass of wheel-free graphs,
because every wheel is a propeller. 

Many interesting classes of graphs can be characterized as being
$\mathcal F$-free for some family $\mathcal F$. The most famous such
example is the class of perfect graphs.  A graph $G$ is {\em perfect}
if for every induced subgraph $H$ of $G$, $\chi (H)=\omega (H)$, where
$\chi (H)$ denotes the {\em chromatic number} of $H$, i.e.\ the minimum
number of colors needed to color the nodes of $H$ so that no two
adjacent nodes receive the same color, and $\omega (H)$ denotes the
size of a largest clique in $H$, where a {\em clique} is a graph in which
every pair of nodes are adjacent.  A {\em hole} in a graph is an
induced cycle of length at least 4.  The famous Strong Perfect Graph
Theorem~\cite{crst}, states that a graph is perfect if and only if it does not
contain an odd hole nor the complement of an odd hole (such graphs are
known as {\em Berge graphs}).  This proof is obtained through a
decomposition theorem for Berge graphs, and in this study wheels and
another set of configurations known as {\em 3-path configurations}
($3PC$'s) play a key role.  The $3PC$'s are structures induced by three
paths $P_1=x_1\ldots y_1$, $P_2=x_2\ldots y_2$ and $P_3=x_3\ldots
y_3$, such that $\{ x_1,x_2,x_3\} \cap \{ y_1,y_2,y_3\}=\emptyset$,
$X=\{ x_1,x_2,x_3\}$ induces either a triangle or a single node, $Y=\{
y_1,y_2,y_3\}$ induces either a triangle or a single node, and the
nodes of $P_i \cup P_j$, $i\neq j$, induce a hole.  More specifically,
a $3PC(\cdot ,\cdot )$ is a $3PC$ in which both $X$ and $Y$ consist of
a single node; a $3PC(\Delta ,\cdot )$ is a $3PC$ in which $X$ induces
a triangle and $Y$ consist of a single node; and a $3PC(\Delta ,\Delta
)$ is a $3PC$ in which both $X$ and $Y$ induce triangles.  It is easy
to see that Berge graphs are both $3PC(\Delta ,\cdot)$-free and
odd-wheel-free (where an {\em odd-wheel} is a wheel that induces an
odd number of triangles).  The remaining wheels and 3PC's form
structures around which the decompositions occur in the decomposition
theorem for Berge graphs in \cite{crst}.  

Wheels and 3PC's are called \emph{Truemper configurations}, and they
play a role in other classes of graphs.  A well studied example is the
class of even-hole-free graphs.  Here again, the
decomposition theorems for this class ~\cite{cckv-ehf1,dsv:ehf} 
are obtained by studying
Truemper configurations that may occur as induced subgraphs.
In both classes (Berge graphs and
even-hole-free graphs), analysing what happens when the graph
contains a wheel is a difficult task.  This suggests that wheel-free
graphs should have interesting structural properties.  This is also
suggested by three subclasses of wheel-free graphs described below.

\begin{itemize}
\item Say that a graph is \emph{unichord-free} if it does not contain
  a cycle with a unique chord as an induced subgraph.  The class of
  unichord-free graphs is a subclass of wheel-free graphs (because
  every wheel contains a cycle with a unique chord as an induced
  subgraph), and unichord-free graphs have a complete structural
  description, see~\cite{nicolas.kristina:one} and also the end of
  Section~\ref{subS:unichord} below.
\item It is easy to see that the class of $K_4$-free graphs that do
  not contain a subdivision of wheel as an induced subgraph is the
  class of graphs that do not contain a wheel or a subdivision of
  $K_4$ as induced subgraphs.  Here again, this subclass of wheel-free
  graphs has a complete structural description,
  see~\cite{nicolas:isk4}.
\item The class of graphs that do not contain a wheel (as a subgraph)
  does not have a complete structural description so far.  However, in
  \cite{thomassenToft:k4} (see also \cite{aboulkerHT:wheelFree}),
  several structural properties for this class are given.  It is also
  proved there that every graph that does not contain a wheel is
  4-colorable, and that every $K_4$-free graph that does not contain a
  wheel is 3-colorable.
\end{itemize}

In Section \ref{sec:dt} we continue this list of well-understood
subclasses of wheel-free graphs by proving decomposition theorems for
graphs that do not contain propellers, both in the subgraph and the
induced subgraph sense.  Based on the decomposition theorems, in
Section~\ref{sec:ra}, we construct polynomial time recognition
algorithms for these two classes of graphs.  Note that the complexity
of detecting a wheel as an induced subgraph is an open question, while
the complexity of detecting the other Truemper's configurations is
settled ($3PC(\Delta, \cdot )$ is polynomial~\cite{cclsv} and is one
of the steps in the polynomial time recognition algorithm for Berge
graphs~\cite{cclsv}, $3PC(\cdot ,\cdot)$ is
polynomial~\cite{cs:3inatree}, $3PC(\Delta, \Delta )$ is
NP-complete~\cite{mt}).  In the same section, we prove that deciding
whether a graph contains, as an induced subgraph, a propeller such
that the center has at least 4 neighbors on the rim is an NP-complete
problem.  It is easy to show directly that propeller-free graphs have
a node of degree at most~2, which implies that the class can be
vertex-colored in polynomial time, see Theorem~\ref{t:vc1}.  In
Section~\ref{sec:extr}, we prove that propeller-free graphs admit what
we call \emph{extreme decompositions}, that are decompositions such
that one of the blocks of decomposition is in some simple basic class to
be defined later.  Using this property we show that 2-connected
propeller-free graphs have an edge both of whose endnodes are of
degree~2.  This property is used to give polynomial time algorithms
for edge-coloring propeller-free graphs.  Observe that since a clique
on four nodes is a propeller, finding the size of a largest clique in
a propeller-free graph can clearly be done in polynomial time. On the
other hand, finding a maximum stable set of a propeller-free graph is
NP-hard (follows easily from \cite{poljak}, see also
\cite{nicolas.kristina:one}).

\subsection*{Terminology and notation}

Let $G$ be a graph.
For $x\in V(G)$, $N(x)$ denotes the set of neighbors of~$x$.
For $S \subseteq V(G)$, $G[S]$ denotes the subgraph
of $G$ induced by $S$, and $G\setminus S=G[V(G)\setminus S]$.
For $x \in V(G)$ we also use notation $G\setminus x$ to denote
$G\setminus \{ x \}$. For $e \in E(G)$, $G\setminus e$ denotes
the graph obtained from $G$ by deleting edge $e$.

A set $S\subseteq V(G)$ is a {\em node cutset} of $G$ if $G \setminus
S$ has more than one connected component.  Note that if $S=\emptyset$,
then $G$ is disconnected.  When $|S|=k$ we say that $S$ is a {\em
  $k$-cutset}.  If $\{ x \}$ is a node cutset of $G$, then we say that
$x$ is a {\em cutnode} of $G$. A 2-cutset $\{ a,b \}$ is a {\em
  $K_2$-cutset} if $ab \in E(G)$, and an {\em $S_2$-cutset} otherwise.
If a graph $G$ has a node cutset $S$, then $V(G)\sm S$ can be
partitioned into two non-empty sets $C_1, C_2$ such that no edge of
$G$ has an end in $C_1$ and an end in $C_2$.  In this situation, we
say that $(S, C_1, C_2)$ is a \emph{split} of $S$.

A {\em path} $P$ is a sequence of distinct nodes $p_1p_2\ldots p_k$,
$k\geq 1$, such that $p_ip_{i+1}$ is an edge for all $1\leq i <k$.
Edges $p_ip_{i+1}$,  for  $1\leq i <k$, are called the {\em edges of $P$}.
Nodes $p_1$ and $p_k$ are the {\em endnodes} of $P$, and 
$p_2\ldots p_{k-1}$ is the {\em interior} of $P$.
$P$ is
refered to as a {\em $p_1p_k$-path}. For two nodes $p_i$ and $p_j$
of $P$, where $j\geq i$, the path $p_i\ldots p_j$ is
called the {\em $p_ip_j$-subpath} of $P$ and is denoted by $P_{p_ip_j}$.
We write $P=p_1\ldots p_{i-1} P_{p_ip_j} p_{j+1} \ldots p_k$ or
$P=p_1\ldots p_{i} P_{p_ip_j} p_{j} \ldots p_k$.
A cycle $C$ is a sequence of nodes $p_1p_2\ldots p_kp_1$, $k \geq 3$,
such that $p_1\ldots p_k$ is a path and $p_1p_k$ is an edge.
Edges $p_ip_{i+1}$,  for  $1\leq i <k$, and edge $p_1p_k$ 
are called the {\em edges of $C$}.
Let $Q$ be  a path or a cycle. 
The node set of $Q$ is denoted by $V(Q)$.
The {\em length} of $Q$ is the number of its edges.
An edge $e=uv$ is a {\em chord} of $Q$ if $u,v\in V(Q)$, but $uv$ is
not an edge of $Q$. A path or a cycle $Q$ in a graph $G$ is {\em chordless}
if no edge of $G$ is a chord of $Q$.

In all complexity analysis of the algorithms, $n$ stands for the number
of nodes of the input graph and $m$ for the number of edges.

\section{Classes defined by constraints on connectivity}\label{sec:mc2cg}

The {\em connectivity} of a graph $G$ is the minimum size of a node
set $S$ such that $G\setminus S$ is disconnected or has only one
node.  A graph is {\em $k$-connected} if its connectivity is at
least $k$.  A graph is \emph{minimally $k$-connected} if it is
$k$-connected and if the removal of any edge yields a graph of
connectivity $k-1$.  A graph is \emph{critically $k$-connected} if it
is $k$-connected and if the removal of any node yields a graph of
connectivity $k-1$.  Minimally and critically $k$-connected graphs
were the object of much research, see~\cite{bollobas:egt} for
instance.  Note that minimally (and critically) $k$-connected graphs
are classes of graphs that are not closed under any classical
containment relation for graphs, such as the subgraph and induced
subgraph containment relations.  But as we shall see, there are
several ways to enlarge a class to make it closed under taking
subgraphs or induced subgraphs.  Here we consider the classes of
minimally and critically 2-connected graphs, and related
hereditary classes that have similar structural properties, but
are algorithmically more convenient to work with.

\subsection{Minimally 2-connected graphs}

In this section we revisit several old results on minimally
2-connected graphs and establish the relationship between this class
and a class that contains it and is closed under taking subgraphs.
Two $xy$-paths $P$ and $Q$ in a graph $G$ are {\em internally
  disjoint} if they have no internal nodes in common, i.e.\ $V(P)
\cap V(Q)=\{ x,y\}$.  We will use the following classical result.

\begin{theorem}[Menger, see \cite{bondy.murty:book}]\label{thm:whitney}
  A graph $G$ on at least two nodes is 2-connected if and only if any
  two nodes of $G$ are connected by at least two internally disjoint
  paths.
\end{theorem}

Let $\mathcal C'_0$ be the class of graphs such that the nodes of
degree at least 3 induce an independent set.  Let $\mathcal C'_1$ be
the class of \emph{chordless graphs}, that are graphs whose cycles are
all chordless (or in other words, the class of graphs that do not
contain a cycle with a chord).  Observe that classes $\mathcal C'_0$
and $\mathcal C'_1$ are both closed under taking subgraphs (and in
particular, they are closed under taking induced subgraphs).  It is
easy to check that ${\mathcal C'_0} \subsetneq {\mathcal C'_1}$.

\begin{lemma}
  \label{th:clM2c}
  A graph $G$ is chordless if and only if for every subgraph $H$ of
  $G$, either $H$ has connectivity at most~1 or $H$ is minimally
  2-connected.
\end{lemma}

\begin{proof}
  A cycle with a chord has connectivity 2 and is not minimally
  2-connected since removing the chord yields a 2-connected graph.
  This proves the ``if'' part of the theorem.  To prove the ``only
  if'' part, consider a chordless graph $G$, and suppose for a
  contradiction that some subgraph $H$ of $G$ is 2-connected and not
  minimally 2-connected.  So by deleting some edge~$e$, a 2-connected
  graph $H'$ is obtained.  By Theorem \ref{thm:whitney}, the two
  endnodes of $e$ are contained in a cycle $C$ of $H'$. But then $C$
  together with $e$ forms in $H$ a cycle with a chord, a
  contradiction.
\end{proof}

Class $\mathcal C'_1$ was studied by Dirac~\cite{dirac:chordless} and
Plummer~\cite{plummer:68} in the 1960s.

\begin{theorem}[Dirac~\cite{dirac:chordless}, Plummer \cite{plummer:68}]
  \label{th:2cCM}
  A 2-connected graph is chordless if and only if it is minimally
  2-connected. 
\end{theorem}

\begin{proof}
  If $G$ is a 2-connected chordless graph, then by
  Lemma~\ref{th:clM2c}, it is minimally 2-connected.  Conversely,
  suppose that $G$ is a minimally 2-connected graph and let $uv$ be an
  edge of $G$.  So, $G\sm uv$ has connectivity~1 and therefore
  contains a cutnode $x$.  Since $G$ is 2-connected, it follows that
  $(G\sm uv)\sm x$ has two connected components, one containing $u$,
  the other containing $v$.  This implies that every cycle of $G$ that
  contains $u$ and $v$ must go through $uv$, so $uv$ cannot be a chord
  of any cycle of $G$.  This proof can be repeated for all edges
  of~$G$.  It follows that $G$ is chordless.
\end{proof}

It seems that it was not observed until recently that the class
$\mathcal C'_1$ of chordless graphs admits a simple decomposition
theorem, with $\mathcal C'_0$ serving as a basic class.  An
$S_2$-cutset $\{a, b\}$ is \emph{proper } if it has a split $(\{a,
b\}, C_1, C_2)$ such that neither $G[\{a, b\} \cup C_1]$ nor $G[\{a,
b\} \cup C_2]$ is a chordless $ab$-path.  When we say that $(\{ a,b \}
,D_1,D_2)$ is a split of a proper $S_2$-cutset, we mean that neither
$G[\{a, b\} \cup D_1]$ nor $G[\{a, b\} \cup D_2]$ is a chordless
$ab$-path.  The following theorem is implicitly proved in
\cite{nicolas.kristina:one} and explicitly stated and proved in
\cite{nicolas:isk4}. We include here a proof that is much shorter and
simpler than the previous ones.

\begin{theorem}
  \label{th:decChLess}
  A graph in $\mathcal C'_1$ is either in $\mathcal C'_0$, or has a
  0-cutset, a 1-cutset, or a proper $S_2$-cutset.
\end{theorem}

\begin{proof}
  Let $G$ be in $\mathcal C'_1 \sm \mathcal C'_0$ and suppose that $G$
  has no 0-cutset and no 1-cutset.  So, in $G$, there is an edge
  $e=uv$ such that $u$ and $v$ have both degree at least 3 and by
  Lemma~\ref{th:clM2c}, $G\sm e$ is not 2-connected so, it has a
  0-cutset (so it is disconnected) or a 1-cutset.

  If $G \setminus e$ is disconnected, then $u$ (and $v$) would be a
  cutnode of $G$. So $G\setminus e$ has a cutnode $w \notin \{u, v\}$.
  Since $w$ is not a cutnode of $G$, the graph $(G\sm e) \sm w$ has
  exactly two connected components $C_u$ and $C_v$, containing $u$ and
  $v$ respectively, and $V(G) = C_u \cup C_v \cup \{w\}$.  Let
  $u'\notin \{v, w\}$ be a neighbor of $u$ ($u'$ exists since $u$ has
  degree at least 3).  So, $u'\in C_u$.  In $G$, $u$ is not a cutnode,
  so there is a path $P_u$ from $u'$ to $w$ whose interior is in
  $C_u \sm \{u\}$.  Together with a path $P_v$ from $v$ to $w$ with interior in
  $C_v$, $P_u$, $uu'$ and $e$ form a cycle, so $uw\notin E(G)$ for
  otherwise $uw$ would be a chord of this cycle.  Because of the
  degrees of $u$ and $v$, $(\{u, w\}, C_u\sm \{u\}, C_v)$ is a split
  of a proper $S_2$-cutset of $G$.
 \end{proof}

Theorem~\ref{th:2cCM} shows that the class of minimally 2-connected
graphs is a subclass of some hereditary class that has a precise
decomposition theorem, namely Theorem~\ref{th:decChLess}.  There is a
more standard way to embed a class $\mathcal C$ into an hereditary
class $\mathcal C'$: taking the closure of $\mathcal C$, that is the
class $\mathcal C'$ of all subgraphs (or induced subgraphs according
to the containment relation under consideration) of graphs from
$\mathcal C$.  But as far as we can see, applying this method to
minimally 2-connected graphs yields a class more difficult to handle
than chordless graphs, as suggested by what follows.  The classes
$\mathcal C_0'$ and $\mathcal C_1'$ are both closed under taking
subgraphs (and in particular under taking induced subgraphs), so the
class of subgraphs of minimally 2-connected graphs is contained in
$\mathcal C_1'$.  On the other hand, a chordless graph, or even a
graph from $\mathcal C'_0$ may fail to be a subgraph of some minimally
2-connected graph.  For instance consider the path on $a, b, c, d$ and
add the edge $bd$.  The obtained graph is chordless, in $\mathcal
C'_0$, and no minimally 2-connected graph may contain it as a
subgraph.  Hence $\mathcal C'_1$ is a proper superclass of the class
of subgraphs of minimally 2-connected graphs.

In the rest of this subsection, we show how Theorem~\ref{th:decChLess}
can be used to prove several known theorems.  The first example is
about edge and total coloring (we do not reproduce the proof, which is
a bit long).  Note that for the proof of the following theorem, the
only approach we are aware of is to use Theorem \ref{th:decChLess}.

\begin{theorem}[Machado, de Figueiredo and
  Trotignon~\cite{mft:chordless}]
  \label{th:mft}
  Let $G$ be a chordless graph of maximum degree at least~3.  Then $G$
  is $\Delta(G)$-edge colourable and $(\Delta(G)+1)$-total-colourable.
\end{theorem}

Dirac~\cite{dirac:chordless} and Plummer~\cite{plummer:68}
independently showed that minimally 2-connected graphs have at least
two nodes of degree at most~2 and chromatic number at most 3.  We
now show how Theorem \ref{th:decChLess} can be used to give simple
proofs of these results for chordless graphs in general.  In the rest
of this subsection, when $(\{ a,b\}, X,Y)$ is a split of a proper
$S_2$-cutset of a graph $G$, we denote by $G_{X}$ the graph obtained
from $G[X \cup \{ a ,b\}]$ by adding a node $y$ that is adjacent to
both $a$ and $b$ ($G_Y$ is defined similarly from $G[Y \cup \{ a,
b\}]$ by adding a node $x$ that is adjacent to both $a$ and $b$).

\begin{theorem}\label{t:chord1}
  Every chordless graph on at least two nodes has at least two
  nodes of degree at most 2.
\end{theorem}

\begin{proof}
  We prove the result by induction on the number of nodes.  If $G
  \in \mathcal C'_0$ then clearly the statement holds.  Let $G\in
  \mathcal C'_1 \setminus \mathcal C_0'$, and assume the statement
  holds for graphs with fewer than $|V(G)|$ nodes.  Suppose $G$ has
  a 0-cutset or 1-cutset $S$, and let $C_1, \ldots ,C_k$ be the
  connected components of $G \setminus S$. For $i=1, \ldots ,k$, by
  induction applied to $G_i=G[V(C_i) \cup S]$, $C_i$ contains a node
  of degree at most 2 in $G_i$. Note that such a node is of degree
  at most 2 in $G$ as well, and hence $G$ has at least two nodes of
  degree at most~2.  So we may assume that $G$ is 2-connected, and
  hence by Theorem~\ref{th:decChLess}, $G$ has a proper $S_2$-cutset
  with split $(\{ a,b\}, X,Y)$. We now show that both $X$ and $Y$
  contain a node of degree at most 2.

  Let $(\{ a',b'\},X',Y')$ be a split of a proper $S_2$-cutset of $G$
  such that $X'\subseteq X$, and out of all such splits assume that
  $|X'|$ is smallest possible. We now show that both $a'$ and $b'$
  have at least two neighbors in $X'$. Since $G$ is 2-connected both
  $a'$ and $b'$ have a neighbor in every connected component of $G
  \setminus \{ a',b'\}$. In particular $G[Y' \cup \{ a',b'\}]$
  contains an $a'b'$-path $Q$ and $a'$ has a neighbor $a_1$ in
  $X'$. Suppose $N(a')\cap X'=\{ a_1\}$. If $a_1b'$ is not an edge,
  then (since $G[X'\cup \{ a',b'\}]$ is not a chordless path), $(\{
  a_1,b'\}, X' \sm \{a_1\}, Y' \cup \{a'\})$ is a split of a proper
  $S_2$-cutset of $G$, contradicting our choice of $(\{ a',b,\}
  ,X',Y')$. So $a_1b'$ is an edge. Then since $G[X'\cup \{ a',b'\}]$
  is not a chordless path, $X'\setminus \{ a_1\}$ contains a node $c$.
  Since $a_1$ cannot be a cutnode of $G$, there is a $b'c$-path in $G
  \setminus a_1$ whose interior nodes are in $X'$.  Since $b'$ cannot
  be a cutnode of $G$, there is an $a_1c$-path in $G \setminus b'$
  whose interior nodes are in $X'$. Therefore $G[X'\cup b']\setminus
  a_1b']$ contains an $a_1b'$-path $P$. But then $V(P)\cup V(Q)$
  induces a cycle with a chord, a contradiction. Therefore, $a'$ has
  at least two neighbors in $X'$ and by symmetry so does $b'$.

  Note that $|V(G_{X'})|<|V(G)|$, and clearly since $G$ is chordless
  so is $G_{X'}$. So by induction, there is a node $t\in
  V(G_{X'})\setminus \{y'\}$ that is of degree at most~2 in
  $G_{X'}$. Since both $a'$ and $b'$ have at least two neighbors in
  $X'$, it follows that $t\in X'$, and hence $t$ is of degree at most
  2 in $G$ as well.  So $X$ contains a node of degree at most 2, and
  by symmetry so does $Y$, and the result holds.
\end{proof}

In the proof above, the key idea to make the induction work is to
consider a split minimizing one of the sides.  This can be avoided by
using a stronger induction hypothesis: in every cycle of a 2-connected
chordless graph that is not a cycle, there exist four nodes $a, b, c,
d$ that appear in this order, and such that $a, c$ have degree 2, and
$b, d$ have degree at least 3.

Note that for proving the theorem below, it is essential that the class we
work on is closed under taking induced subgraphs.  This is why proofs
of 3-colorability in \cite{dirac:chordless} and \cite{plummer:68} are
more complicated (they consider only minimally 2-connected graphs,
that are not closed under taking subgraphs).

\begin{corollary}
If $G$ is a chordless graph then $\chi (G)\leq 3$.
\end{corollary}
\begin{proof}
  Let $G$ be a chordless graph and by Theorem \ref{t:chord1} let $x$
  be a node of $G$ of degree at most 2. Inductively color $G\setminus
  x$ with at most 3 colors. This coloring can be extended to a
  3-coloring of $G$ since $x$ has at most two neighbors in $G$.
\end{proof}

We now show how Theorem~\ref{th:decChLess} may be used to prove the
main result in \cite{plummer:68}, that is Theorem~\ref{t:plummer}
below.  We need the next lemma whose simple proof is omitted.

\begin{lemma}[see \cite{mft:chordless}]
  \label{l:extremal}
  Let $G$ be a 2-connected  chordless graph not in ${\cal C}'_0$.  Let $(X, Y, a,
  b)$ be a split of a $S_2$-cutset of $G$ such that $|X|$ is minimum among
  all possible such splits.  Then $G_X$ is in ${\cal C}'_0$.  Moreover, $a$ and $b$
  both have degree at least~3 in $G$ and in $G_X$.
\end{lemma}

\begin{theorem}[Plummer \cite{plummer:68}]
  \label{t:plummer}
  Let $G$ be  a 2-connected graph.  Then $G$ is minimally 2-connected
  if and only if either
  \begin{enumerate}
    \item\label{i1} $G$ is a cycle; or 
    \item\label{i2} if $S$ denotes the set of nodes of degree 2 in
      $G$, then there are at least two components in $G\sm S$, each
      component of $G\sm S$ is a tree and if $C$ is any cycle in $G$
      and $T$ is any component of $G\sm S$, then $(V(C) \cap V(T),
      E(C) \cap E(T))$ is empty or connected.
  \end{enumerate}
\end{theorem}

\begin{proof}
  Suppose first that $G$ is minimally 2-connected (or equivalently
  chordless).  If $G$ is in ${\cal C}'_0$ then $G\sm S$ contains only
  isolated nodes.  Hence, either $G\sm S$ is empty, in which case
  all nodes of $G$ are of degree 2, meaning that $G$ is a cycle; or
  $G\sm S$ is not empty, in which case $G$ contains at least two
  nodes of degree at least 3, and the second outcome holds.

  So, by Theorem~\ref{th:decChLess}, we may assume that $G$ admits a
  proper $S_2$-cutset $\{a, b\}$.  By Lemma~\ref{l:extremal}, we consider
   $G_X$ and $G_Y$ so that $G_X$ is in ${\cal C}'_0$ and $a, b$ have
  degree 3 in $G_X$.  Note that from the definition of a proper
  $S_2$-cutset, none of $G_X, G_Y$ is a cycle.

  Inductively, let $S_Y$ be the set of nodes of degree 2 in $G_Y$ and
  let $T_1, \dots, T_k$ be the components of $G_Y\sm S_Y$.  If $a$ or
  $b$ has degree~2 in $G_Y$, then, it has neighbors in at most one of
  the $T_i$'s (and in fact has a unique neighbor in it).  So, if for
  some $i\in \{1, \dots, k\}$, $T_i$ does not contain $a$ (resp.\ $b$,
  resp.\ $a, b$), and if $T_i$ is linked by some edge to $a$ (resp.\
  $b$, resp.\ $a$ and $b$), then we define the tree $T'_i$ to be the
  tree obtained by adding the pendent node $a$ (resp.\ $b$, resp.\
  both $a$ and $b$) to $T_i$.  For all $j=1, \dots, k$ such that
  $T'_j$ is not defined above, we put $T'_j=T_j$.  Now, if we remove
  the nodes of degree 2 of $G$, $T'_1$, \dots, $T'_k$ are connected
  components (here we use the fact that since $G_X$ is in ${\cal
    C}'_0$, all neighbor of $a$ or $b$ in $X$ have degree 2).  The
  other components are the nodes of degree at least~3 from $X$.  They
  are all trees because $G_X$ is in ${\cal C}'_0$.

  It remains to prove that if $C$ is any cycle in $G$ and $T$ is any
  component of $G\sm S$, then $(V(C) \cap V(T), E(C) \cap E(T))$ is
  empty or connected.  Let $C$ be a cycle of~$G$.  There are three
  cases.  Either $V(C) \subseteq X\cup \{a, b\}$, or $V(C) \subseteq
  Y\cup \{a, b\}$, or $C$ is formed of a path $P_X$ from $a$ to $b$
  with interior in $X$ and a path $P_Y$ from $a$ to $b$ with interior
  in $Y$.  In the first case, the trees intersected by $C$ are all
  formed of one node, so~\ref{i2} holds.  In the second case, $C$ is
  also a cycle of $G_Y$.  Let $T$ be a tree of $G\sm S$ such that
  $V(T) \subset Y \cup \{a, b\}$ (all the other trees of $G\sm S$ are
  on 1 node).  Note that $a\in V(C) \cap V(T)$ implies that $a$ has
  degree at least 3 in $G_Y$ and so $T$ is also a tree of $G_Y \sm
  S_Y$.  Hence, $(V(C) \cap V(T), E(C) \cap E(T))$ is connected by the
  induction hypothesys applied to $G_Y$.  In the third case, we
  consider the cycle $C_Y$ formed by $P_Y$ and the marker node of
  $G_Y$.  We suppose that $T$ has more than one node (otherwise the
  proof is easy), so $V(T) \subseteq Y \cup \{a, b\}$.  Note that if
  $T$ goes through $a$, then it must go through some neighbor of $a$
  in $Y$.  This means that if $a$ has degree 2 in $G_Y$ and $a \in
  V(C) \cap V(T)$, then the neighbor $a'$ of $a$ in $G_Y$ has degree
  at least~3 and is therefore in a tree of $G_Y \sm S_Y$, so $a' \in
  V(C) \cap V(T)$.  The same remark holds for $b$.  Hence, $(V(C) \cap
  V(T), E(C) \cap E(T))$ is connected by the induction hypothesys
  applied to $G_Y$.

  Suppose conversly that one of \ref{i1}, \ref{i2} is satisfied by
  some 2-connected graph $G$ (here we reproduce the proof given by
  Plummer).  If $G$ is a cycle, then it is obviously minimally
  2-connected.  Otherwise, let $e=uv$ be an edge of $G$.  If it enough
  to prove that $G\sm e$ is not 2-connected.  If $u$ or $v$ has
  degree~2 in $G$ this holds obviously.  Otherwise, $u$ and $v$ are in
  the same component $T$ of $G\sm S$.  If $G\sm e$ is 2-connected,
  then some cycle $C$ of $G\sm e$ goes through $u$ and $v$, and $(V(C)
  \cap V(T), E(C) \cap E(T))$ is not connected nor empty because it
  contains $u$ and  $v$ but not $e=uv$ (and removing any edge from a
  tree disconnects it), a contradiction to~\ref{i2}.
\end{proof}

Note that we do not use the existence of nodes of degree 2 to prove
the theorem above. Hence, a new proof of their existence can be given:
if $G$ is 2-connected, then by Theorem~\ref{t:plummer}, the nodes
of degree 2 of $G$ form a cutset of $G$.  Hence, there must be at
least two of them; otherwise, the existence of two nodes of degree at
most~2 follows easily by induction.

\label{subS:unichord}We close this subsection by observing that there is
another well studied hereditary class that properly contains the class
$\mathcal C_1'$, namely the class of graphs that do not contain a
cycle with a unique chord as an induced subgraph.  In
\cite{nicolas.kristina:one}, a precise structural description of this
class is given and used to obtain efficient recognition and coloring
algorithms.  Interestingly, it was proved by McKee~\cite{mcKee:indep}
that these graphs can be defined by constraints on connectivity: the
graphs with no cycles with a unique chord are exactly the graphs such
that all minimal separators are independent sets (where a
\emph{separator} in a graph $G$ is a set $S$ of nodes such that $G\sm S$
has more connected components than $G$).

\subsection{Critically 2-connected graphs}\label{secc2c}

In this subsection we consider the class of critically 2-connected
graphs, that were studied by Nebesk\'y~\cite{nebesky:c2c}, and
investigate whether there exists an analogous sequence of theorems as
in the previous subsection, starting with critically $2$-connected
graphs instead of minimally 2-connected graphs.  An analogue of
Lemma~\ref{th:clM2c} exists with ``critically'' instead of
``minimally'' and ``propeller'' instead of ``cycle with a chord''.

\begin{lemma} \label{th:c2c}
  A graph $G$ does not contain a propeller  if and only if for
  every subgraph $H$ of $G$, either $H$ has connectivity at most~1 or
  it is critically 2-connected.
\end{lemma}

\begin{proof}
  A propeller has connectivity 2 and is not critically 2-connected
  since removing the center yields a 2-connected graph.  This proves
  the ``if'' part of the theorem.  To prove the ``only if'' part,
  consider a graph $G$ that contains no propeller, and suppose for a
  contradiction that some subgraph $H$ of $G$ does not satisfy the
  requirement on connectivity that is to be proved.  Hence $H$ is
  2-connected and not critically 2-connected.  So by deleting a node
  $v$, a 2-connected graph $H'$ is obtained. Note that $|V(H)|\geq
  4$. Since $v$ has at least two neighbors $u$ and $w$ in $H'$
  (because of the connectivity of $H$), by Theorem~\ref{thm:whitney},
  $H'$ contains a cycle $C$ through $u$ and $w$, and $(C,u)$ is a
  propeller of $H$, a contradiction.
\end{proof}

An anologue of Theorem~\ref{th:2cCM} seems hopeless. A critically
2-connected graph can contain anything as a subgraph: the class of the
subgraphs of critically 2-connected graphs is the class of all graphs.
To see this, consider a graph $G$ on $\{v_1, \dots, v_n\}$.  If $G$ is
not connected, then add a node $v_{n+1}$ adjacent to all nodes.  For
every node $v_i$, add a node $a_i$ adjacent to $v_i$ and a node $b_i$
adjacent to $a_i$.  Add a node $c$ adjacent to all $b_i$'s.  It is
easy to see that the obtained graph is critically 2-connected, and
contains $G$ as a subgraph.  So there cannot be a version of
Theorem~\ref{th:2cCM} with ``critically'' instead of ``minimally'': a
critically 2-connected graph may contain a propeller, since it may
contain anything.  Also, any property of graphs closed under taking
subgraphs, such as being $k$-colorable, is false for critically
2-connected graphs, unless it holds for all graphs.  However, there is
a sequence of theorems, proven here, that mimics the sequence obtained
by thinking of minimally 2-connected graphs.  Note that containing a
cycle with a chord as a subgraph is equivalent to containing a cycle
with a chord as an induced subgraph; while containing a propeller as a
subgraph is not equivalent to containing a propeller as an induced
subgraph.  So, there are two ways to find an analogue of chordless
graphs, and in this paper we consider both.

Let $\mathcal C_0$ be the class of graphs with no node having at least
two neighbors of degree at least three.  Let $\mathcal C_1$ be the
class of graphs that do not contain a propeller.  Let $\mathcal C_2$
be the class of graphs that do not contain a propeller as an induced
subgraph.  It is is easy to check that $\mathcal C_0\subsetneq
\mathcal C_1 \subsetneq \mathcal C_2$.

Before studying decomposition theorems for $\mathcal C_1$ and
$\mathcal C_2$, let us see that an analogue of Theorem~\ref{t:chord1}
can be proved directly for propeller-free graphs (and implies that
they are 3-colorable).  Nebesk\'y~\cite{nebesky:c2c} proved that every
critically 2-connected graph contains a node of degree~2, but
critically 2-connected graphs are not 3-colorable in general, since
they may contain any subgraph of arbitrarily large chromatic number.
Note that studying longest paths to obtain nodes of small degree in
graphs where ``propeller-like'' structures are excluded can give much
stronger results, see~\cite{turner:05}.

\begin{theorem}\label{t:vc1}
  If $G\in \mathcal C_2$, then $G$ has a node of degree at most 2 and
  $G$ is 3-colorable.
\end{theorem}
\begin{proof}
  Suppose that for every $v\in V(G)$, $d(v) \ge 3$.  Let $P$ be a
  longest chordless path in $G$, and $x$ and $y$ the endnodes of $P$.
  As $d(x) \geq 3$, $x$ has at least two neighbors $u$ and $v$ not in
  $P$ and $u$ (resp.\ $v$) has a neighbor in $P\setminus x$, since
  otherwise $V(P) \cup \{ u\}$ (resp.\ $V(P) \cup \{ v\}$) would
  induce a longer path in $G$. We choose $u_1$ and $v_1$, neighbors of
  respectively $u$ and $v$ in $P\setminus x$ that are closest to $x$
  on $P$.  W.l.o.g.\ let us assume that $x, u_1, v_1$ appear in this
  order on $P$. Then $( v P_{xv_1} v , u)$ is an induced propeller of
  $G$, a contradiction.  This proves that $G$ has a node of degree
  at most~2.  It follows by an easy induction that every graph from
  $\mathcal C_2$ is 3-colorable.
\end{proof}

\section{Decomposition theorems}\label{sec:dt}

In this section we present decomposition theorems for graphs that
do not contain propellers and graphs that do not contain propellers
as induced subgraphs.

A $K_2$-cutset $S$ of a graph $G$ is \emph{proper} if $G \sm S$
contains no node adjacent to all nodes of $S$.

\begin{lemma}\label{l:pK2cutset}
  If $G$ is a 2-connected graph from $\mathcal C_2$, then every
  $K_2$-cutset of $G$ is proper.
\end{lemma}
\begin{proof}
  Let $(\{ a,b \}, A,B)$ be a split of a $K_2$-cutset that is not
  proper.  W.l.o.g.\ $A$ contains a node $x$ that is adjacent to both
  $a$ and $b$.  Since $G$ is 2-connected, both $a$ and $b$ have a
  neighbor in the same connected component of $G[B]$, and hence $G[\{
  a,b\} \cup B]$ contains a chordless cycle $C$ passing through edge
  $ab$.  But then $(C,x)$ is a propeller that is contained in $G$ as
  an induced subgraph, a contradiction.
\end{proof}

\begin{theorem}\label{decompositionC_1}
  A graph in $\mathcal C_1$ is either in $\mathcal C_0$ or it has a
  0-cutset, a 1-cutset, a proper $K_2$-cutset or a proper $S_2$-cutset.
\end{theorem}

\begin{proof}
  Let $G$ be a 2-connected graph in $\mathcal C_1 \sm \mathcal C_0$.  
So $G$ contains
  a node $w$ that has two neighbors $u$ and $v$ that are both
  of degree at least 3. Suppose $uv\in E(G)$ and let $u'\notin
  \{v, w\}$ be a neighbor of $u$.  Since $u$ cannot be a cutnode, 
  there is a path $P$ from $u'$ to $\{v, w\}$ in $G\sm u$ and hence $G[V(P)
  \cup \{u, v, w\}]$ contains a propeller, a contradiction.  
So $uv \notin
  E(G)$.

  If no node of $G\sm w$ is a cutnode separating $u$ from $v$, then
  by Theorem \ref{thm:whitney}, there is a cycle of $G\sm w$ going through
  $u$ and $v$, so that in $G$, $w$ is the center of a propeller, a
  contradiction.  Hence there is such a cutnode $w'$.  So, in
  $G\sm\{w, w'\}$, there are distinct components $C_u$ and $C_v$
  containing $u$ and $v$ respectively, and possibly other components
  whose union is denoted by $C$. 
But then $(\{w, w'\}, C\cup C_u, C_v)$ is a split of either an
$S_2$-cutset
  of $G$ (when $ww'\notin E(G)$), which is proper  
because of the degrees of $u$ and $v$, 
or a split of a $K_2$-cutset (when
  $ww'\in E(G)$), which is proper by Lemma \ref{l:pK2cutset}.  
\end{proof}

A 3-cutset $\{ u,v,w \}$
of a graph $G$ is an
\emph{$I$-cutset}  if the following hold.
\begin{itemize}
\item $G[\{ u,v,w \}]$ contains exactly one edge.
\item There is a partition $(\{u, v, w\}, K',
  K'')$ of $V(G)$ such that:
\begin{enumerate}
\item no edge of $G$ has an endnode in $K'$ and an endnode in $K''$;
\item\label{i:2} 
for some connected component $C'$ of $G[K']$,
$u,v$ and $w$ all have a neighbor in $C'$; and
\item\label{i:3} 
for some connected component $C''$ of $G[K'']$,
$u,v$ and $w$ all have a neighbor in $C''$.
\end{enumerate}
\end{itemize}

In these circumstances, we say that $(\{ u, v, w\}, K', K'')$ is a \emph{split}
of the $I$-cutset $\{ u,v,w\}$.

\begin{theorem}\label{decompositionC_2}
  If a graph $G$ is in $\mathcal C_2$, then either $G\in
  \mathcal C_1$ or $G$ has an $I$-cutset.
\end{theorem}
\begin{proof}
  Let $G$ be a graph in $\mathcal C_2\setminus \mathcal
  C_1$, and let $(C,x)$ be a propeller of $G$ whose rim has the fewest
  number of chords. Note that $C$ must have at least one chord.

\vspace{2ex}

\noindent{\bf Claim 1:} {\em Let $y'y''$ be a chord of $C$, and $P_1$ and 
$P_2$ the two $y'y''$-subpaths of $C$. If a node $u\in V(G)\setminus V(C)$ 
has more than one neighbor on $C$, then it has exactly two neighbors on $C$, 
one in the interior of $P_1$, and the other in the interior of $P_2$.}

\vspace{2ex}

\noindent{\em Proof of Claim 1:} Let $u\in V(G)\setminus V(C)$ and 
suppose that $u$ has at least two neighbors on $C$. If $u$ has at least 
two neighbors on $P_i$, for some $i\in\{1,2\}$, then $G[V(P_i)\cup\{u\}]$ 
contains a propeller that contradicts our choice of $(C,x)$. 
This completes the proof of Claim 1.

\vspace{2ex}

By Claim 1, $x$ has exactly two neighbors $x'$ and $x''$ on $C$.

\vspace{2ex}

\noindent{\bf Claim 2:} {\em If $u\in V(G)\setminus(V(C)\cup\{x\})$ 
then $u$ has at most one neighbor on $C$.}

\vspace{2ex}

\noindent{\em Proof of Claim 2:} Assume not. Then by Claim 1, $u$ has
exactly two neighbors $u'$ and $u''$ on $C$. Let $P_1$ and $P_2$ be
the two $u'u''$-subpaths of $C$. Note that since $C$ has a chord, by
Claim 1 that chord has one endnode in the interior of $P_1$ and the
other in the interior of $P_2$. In particular, neither $P_1$ nor $P_2$
is an edge. If $\{x',x''\}\subset V(P_i)$, for some $i\in\{1,2\}$,
then the graph induced by $V(P_i)\cup\{u,x\}$ contains a propeller
with center $x$ that contradicts our choice of $(C,x)$. So
w.l.o.g.\ $x'$ is contained in the interior of $P_1$ and $x''$ in the
interior of $P_2$. Let $y'y''$ be a chord of $C$. Then by Claim 1 we
may assume that nodes $u'$, $x'$, $y'$, $u''$, $x''$, $y''$ 
are all distinct and appear in this
order when traversing $C$ clockwise.
If $u'y''$ is an edge then the graph induced by
$V(P_1)\cup\{u,y''\}$ contains a propeller with center $y''$
that contradicts our choice of $(C,x)$. So $u'y''$ is not an edge, and
by symmetry neither is $u''y'$. Let $P_1'$ (respectively $P_2'$) be
the $u'y'$-subpath (respectively $u''y''$-subpath) of $C$ that
contains $x'$ (respectively $x''$). Then the graph induced by
$V(P_1')\cup V(P_2')\cup\{u,x\}$  contains a
propeller with center $x$ that contradicts our choice of
$(C,x)$. This completes the proof of Claim 2.

\vspace{2ex}

Let $y'y''$ be a chord of $C$. 
By Claim 1, nodes $x'$, $y'$, $x''$, $y''$ are all
distinct and w.l.o.g.\ appear in this order when traversing $C$
clockwise. Let $P'$ (respectively $P''$) be the $y'y''$-subpath of $C$
that contains $x'$ (respectively $x''$).

\vspace{2ex}

\noindent{\bf Claim 3:} {\em $C$ cannot have a chord $z'z''$ such that
$z' \in V(P') \setminus \{ y',y''\}$ and 
$z'' \in V(P'') \setminus \{ y',y''\}$.}

\vspace{2ex}

\noindent{\em Proof of Claim 3:}
Assume it does. W.l.o.g.\ $z'$ is on the $x'y'$-subpath of $P'$.
Then, by Claim 1, $z''$ is on the 
$x''y''$-subpath of $P''$. Let $C'$ be the cycle obtained by following
$P'$ from $z'$ to $y''$, going along edge $y''y'$, following $P''$ from
$y'$ to $z''$, and going along edge $z''z'$. Since $C'$ cannot have fewer
chords than $C$ (by the choice of $(C,x)$), it follows that both
$z'y'$ and $z''y''$ are edges. But then $G[V(P') \cup \{ z''\}]$
contains a propeller with center $y'$, that contradicts our
choice of $(C,x)$.
This completes the proof of Claim 3.

\vspace{2ex}
  
Assume that $S$ is not a cutset of $G$ that separates $x'$ from
$x''$. Then there exists a shortest path $P = p_1p_2 \dots p_k$ in
$G\sm S$ such that $p_1$ has a neighbor $u \in P'\sm \{y',y''\}$ and
$p_k$ has a neighbor $v \in P''\sm \{y',y''\}$. Finding a
contradiction will complete the proof, since Conditions~\ref{i:2}
and~\ref{i:3} in the definition of an $I$-cutset are satisfied because
of $C$ and $x$.  By Claim~3, $P$ has length at least~2.  By Claim~2
and the definition of $P$, $P$ is a chordless path, $N(p_1) \cap V(C)
= \{u\}$, $N(p_k) \cap V(C) = \{v\}$, and the only nodes of $(C, x)$
that may have a neighbor in the interior of $P$ are $x$, $y'$, and
$y''$.

Let $P_{uy''}$ (respectively $P_{y''v}$) be the $uy''$-subpath 
(respectively $y''v$-subpath) of $C$ that does not contain $y'$. 
Let $P_{uy'}$ (respectively $P_{y'v}$) be the $uy'$-subpath 
(respectively $y'v$-subpath) of $C$ that does not contain $y''$.

\vspace{2ex}

\noindent{\bf Claim 4:} {\em $y'$ and $y''$ have no neighbors in $P$.}

\vspace{2ex}

\noindent{\em Proof of Claim 4:} First suppose that both $y'$ and $y''$ have 
a neighbor in $P$. Let $p_i$ (respectively $p_j$) be the node of $P$ with 
smallest index adjacent to $y'$ (respectively $y''$). W.l.o.g.\ $i\leq j$. 
Let $Q$ be a chordless path from $u$ to $y''$ in $G[V(P_{uy''})]$. 
Then $V(Q)\cup\{p_1,p_2,\ldots,p_j,y'\}$ induces in $G$ a propeller 
with center $y'$, a contradiction.

So we may assume w.l.o.g that $y''$ does not have a neighbor in $P$. 
Suppose $y'$ does. Let $Q$ be the $uv$-subpath of $C$ that contains $y''$. 
Let $Q'$ be a chordless $uv$-path in $G[V(Q)]$. By Claim 3, $Q'$ contains
$y''$.
But then $G[V(Q')\cup V(P)\cup\{y'\}]$ is a propeller with center $y'$,
a contradiction. This completes the proof of Claim 4.

\vspace{2ex}

By symmetry 
it suffices to consider the following two cases.

\vspace{2ex}

\noindent
{\bf Case 1:} $x'\in V(P_{uy''})$ and $x''\in V(P_{y''v})$.

\noindent Let $C'$ be the cycle that consists of $P_{uy''}$, $P_{y''v}$ 
and $P$. Then by Claim 4 $(C',x)$ is a propeller that contradicts our choice 
of $(C,x)$.

\vspace{2ex}

\noindent
{\bf Case 2:} $x'\in V(P_{uy'})$ and $x''\in V(P_{y''v})$.

\noindent Suppose $y'v$ is an edge. Let $C'$ be the cycle that
consists of $P_{uy''}$, $P_{y''v}$ and $P$. Then by Claim~4 $(C',y')$
is a propeller that contradicts our choice of $(C,x)$. So $y'v$ is not
an edge, and by symmetry neither is $uy''$. Now let $C'$ be the cycle
that consists of $P_{uy'}$, $y'y''$, $P_{y''v}$ and $P$. Then by
Claim~4, $(C',x)$ is a propeller that contradicts our choice of
$(C,x)$.
\end{proof}

\section{Recognition algorithms}\label{sec:ra}

Deciding whether a graph contains a propeller can be done directly
as follows: for every 3-node path $xyz$, check whether there are
two internally disjoint $xz$-paths in $G \setminus y$. Since
checking whether there are two internally disjoint $xz$-paths
can be done in $\mathcal O(n)$ time (\cite{ni}, see also \cite{schrijver}),
this leads to an $\mathcal O(n^4)$ recognition algorithm for class
$\mathcal C_1$.

Recognizing whether a graph contains a propeller as an induced
subgraph is a more difficult problem, and we are not aware of any
direct method for doing that. Observe that the above method would not
work since checking whether there is a chordless cycle through two
specified nodes of an input graph is NP-complete~\cite{bien}.  In
Section~\ref{ss:d4p}, we give an NP-completeness result showing that
the detection of ``propeller-like'' induced subgraph may be hard.  In
Section~\ref{ss:raC}, an $\mathcal O(nm)$ decomposition based
recognition algorithm for $\mathcal C_1$ (using Theorem
\ref{decompositionC_1}) is given.  In Section~\ref{ss:rac2}, an
$\mathcal O(n^2m^2)$ decomposition based recognition algorithm for
$\mathcal C_2$ (using Theorem \ref{decompositionC_2}) is given.

\subsection{Detecting 4-propellers}
\label{ss:d4p}

A \emph{4-propeller} is a propeller whose center has at least four
neighbors on the rim.

\begin{theorem}
  The problem whose instance is a graph $G$ and whose question is
  ``does $G$ contain a 4-propeller as an induced subgraph?'' is
  NP-complete.
\end{theorem}

\begin{proof}
  Let $H$ be a graph of maximum degree 3, with 2 non-adjacent nodes
  $x$ and $y$ of degree 2.  Detecting an induced cycle through $x$ and
  $y$ in $H$ is an NP-complete problem (see Theorem~2.7
  in~\cite{leveque.lmt:detect}).  We now show how to reduce this
  problem to the detection of a 4-propeller.  Let $x'$ and $x''$
  (resp.\ $y'$ and $y''$) be the neighbors of $x$ (resp.\ of $y$).
  Subdivide the edges $xx'$, $xx''$, $yy'$ and $yy''$.  Call $a$, $b$,
  $c$, $d$ the four nodes created by these subdivisions.  Add a
  node $v$ adjacent to $a$, $b$, $c$ and $d$.  Call $G$ this new
  graph.  Note that since $H$ has maximum degree~3, $v$ is the only
  node of degree at least~4 in $G$, so every 4-propeller of $G$ must
  be centered at~$v$.  Hence, $G$ contains a 4-propeller if and only
  if $H$ contains an induced cycle through $x$ and $y$.
\end{proof}

Note that detecting (as an induced subgraph) a propeller whose center
has exactly two neighbors on the rim is mentioned
in~\cite{leveque.lmt:detect} as an open problem (Section~3.3, the
first of the 7 open problems).

\subsection{Recognition algorithm for $\mathcal C_1$}
\label{ss:raC}

We first define blocks of decomposition w.r.t.\ different cutsets.

If $G$ has a 0-cutset, i.e.\ it is disconnected, 
then its {\em blocks of decomposition} are the
connected components of $G$.
If $G$ has a 1-cutset $\{ u \}$ and $C_1, \ldots ,C_k$ are the connected
components of $G \setminus u$, then the {\em blocks of decomposition} 
w.r.t.\ this cutset are graphs $G_i=G[C_i \cup \{ u \}]$, for 
$i=1, \ldots ,k$.

Let $(S, A, B)$ be a split of a proper $K_2$-cutset of $G$.  The {\em
  blocks of decomposition} of $G$ with respect to this split are graphs
$G'= G[S\cup A]$ and $G''=G[S\cup B]$.

Let $(\{u,v\}, K', K'')$ be a split of a proper $S_2$-cutset of $G$.  The
{\em blocks of decomposition} of $G$ with respect to this split are
graphs $G'$ and $G''$ defined as follows.  Block $G'$ is the graph
obtained from $G[V(K')\cup\{u,v\}]$ by adding new nodes $u'$ and $v'$,
and edges $uu'$, $u'v'$ and $v'v$.  Block $G''$ is the graph obtained
from $G[V(K'')\cup\{u,v\}]$ by adding new nodes $u''$ and $v''$, and
edges $uu''$, $u''v''$ and $v''v$.  Nodes $u', v', u'', v''$ are
called the \emph{marker nodes} of their block of decomposition.

\begin{lemma}\label{blocksC_1}
For 0-cutsets, 1-cutsets and proper $K_2$-cutsets the following holds:
 $G$ is in $\mathcal C_1$  (resp.\ $\mathcal C_2$) 
if and only if all the blocks of decomposition are in $\mathcal C_1$ 
(resp.\ $\mathcal C_2$).
\end{lemma}

\begin{proof}
Since a propeller is 2-connected, the theorem obviously holds for
0-cutsets and 1-cutsets. Suppose that $(\{ u,v \} ,K',K'')$ is a split
of a proper $K_2$-cutset of $G$, and let $G'$ and $G''$ be the blocks
of decomposition w.r.t.\ this split. Since $G'$ and $G''$ are induced
subgraphs of $G$, it follows that if $G \in \mathcal C_1$
(resp.\ $\mathcal C_2$),
then $G',G'' \in \mathcal C_1$ (resp.\ $\mathcal C_2$). 
If $G'$ and $G''$ are in $\mathcal C_2$, then clearly (since $\{ u,v \}$
is proper) $G$ is in $\mathcal C_2$.
Finally assume that 
$G'$ and $G''$ are in $\mathcal C_1$, but that a propeller $(C,x)$
is a subgraph of $G$. Since $\{ u,v \}$ is proper, it follows that
$C$ contains a node of $K'$ and a node of $K''$. Hence, $C$ contains
both $u$ and $v$, and so w.l.o.g.\ $x$ is in $K'$. But then the
$uv$-subpath of $C$ whose interior nodes are in $K'$, together with edge
$uv$ and node $x$, induces a propeller that is a subgraph of $G'$,
a contradiction.
\end{proof}

\begin{lemma}\label{blocksC_12}
Let $G$ be a 2-connected graph that does not have a proper $K_2$-cutset.
  Let  $(\{u,v\}, K', K'')$ be a split of a proper $S_2$-cutset of $G$, 
  and let $G'$ and $G''$ be the 
blocks of decomposition w.r.t.\ this split.
Then $G\in \mathcal C_1$ if and only if $G'\in \mathcal C_1$
and $G''\in \mathcal C_1$.
\end{lemma}
\begin{proof}
  Assume that $G \in \mathcal C_1$ and w.l.o.g.\ that a propeller
  $(C,x)$ is a subgraph of $G'$. Since $(C,x)$ is not a subgraph of
  $G$, $(C,x)$ must contain at least one of the marker nodes $u'$ or $v'$. By
  the definition of $u'$ and $v'$ it is clear that $u',v'\in V(C)$.  
Since $G$ has
  no 1-cutset both $u$ and $v$ have a neighbor in every connected
  component of $G\setminus\{u,v\}$, and hence $G[V(K'')\cup\{u,v\}]$
  contains a path $P$ from $u$ to $v$. If in $C$ we replace path
  $uu'v'v$ with $P$, we get a propeller $(C',x)$, which is a
  subgraph of $G$, a contradiction.

  To prove the converse assume that $G'\in\mathcal C_1$ and
  $G''\in\mathcal C_1$, but that a propeller $(C,x)$ is a subgraph of
  $G$. Since $(C,x)$ cannot be a subgraph of $G'$ or $G''$, $(C,x)$
  must contain nodes from both $K'$ and $K''$. Then clearly
  $x\not\in\{u,v\}$, so w.l.o.g.\ we may assume that $x\in K'$. If
  there is a node from $V(C)$ in $K'$, then there are $uv$-paths
  $P_1$ and $P_2$ in $G[V(K')\cup\{u,v\}]$ and $G[V(K'')\cup\{u,v\}]$,
  respectively, such that $V(P_1)\cup V(P_2)=V(C)$. Replacing in $C$
  path $P_2$ with $uu'v'v$, we get a propeller $(C',x)$ which is a
  subgraph of $G'$, a contradiction. Therefore, $C$ is contained in
  $G[V(K'')\cup\{u,v\}]$. Since $x\in V(K')$, it has no neighbor in
  $K''$. Since $(C,x)$ is a propeller, it follows that $C$ contains
  both $u$ and $v$, and $x$ is adjacent to both $u$ and $v$. By
  definition of the proper $S_2$-cutsets, $G[V(K')\cup\{u,v\}]$ is not a
  path, and hence $K'$ must contain a node $y$, distinct from
  $x$. 
If there is a node $w$ adjacent to both $x$ and $v$, then $w$ is in $K'$
and hence $\{ u,v,u',v',x,w\}$ induces a propeller in $G'$ with center $w$.
Therefore, no node of $G$ is adjacent to both $x$ and $v$, and by symmetry,
no node of $G$ is adjacent to both $x$ and $u$.
Since $\{x,v\}$ cannot be a proper $K_2$-cutset in $G$, there is a path
  from $u$ to $y$ in $G[(V(K')\cup\{u\})\setminus\{x\}]$. Similarly,
  since $\{x,u\}$ cannot be a proper $K_2$-cutset in $G$, 
there is a path from $y$
  to $v$ in $G[(V(K')\cup\{v\})\setminus\{x\}]$. Therefore, there is a
  path $P$ from $u$ to $v$ in
  $G[(V(K')\cup\{u,v\})\setminus\{x\}]$. But then
  $V(P)\cup\{u',v',x\}$ induces a graph in $G'$ that contains a propeller
  with center $x$, a contradiction.
\end{proof}

\begin{theorem}\label{a:r1}
There is an algorithm with the following specifications.
\begin{description}
\item[ Input:] A graph $G$.
\item[ Output:]
$G$ is correctly identified as not belonging to $\mathcal C_2$, 
or a list $\mathcal L$ of induced subgraphs of $G$ such that:
\begin{itemize}
\item[(i)] $G\in \mathcal C_1$ if and only if for every $L \in \mathcal L$,
$L \in \mathcal C_1$;
\item[(ii)] $G\in \mathcal C_2$ if and only if for every $L \in \mathcal L$,
$L \in \mathcal C_2$;
\item[(iii)] for every $L \in \mathcal L$, $L$ is 2-connected and does not
have a $K_2$-cutset; 
\item[(iv)] $\sum_{L\in \mathcal L} |V(L)| \leq 6n$ and
$\sum_{L\in \mathcal L} |E(L)| \leq 2n+m$
\end{itemize}
\item[ Running time:] $\mathcal O(nm)$
\end{description}
\end{theorem}

\begin{proof}
Consider the following algorithm.
\begin{description}
\item[ Step 1:]
Let $\mathcal L =\mathcal F =\emptyset$. 

\item[ Step 2:]
Find maximal 2-connected components of $G$ 
(i.e.\ decompose $G$ using 0-cutsets and 1-cutsets)
and add them to $\mathcal F$.

\item[ Step 3:]
If $\mathcal F =\emptyset$ then return $\mathcal L$ and stop.
Otherwise, remove a graph $F$ from $\mathcal F$.

\item[ Step 4:]
Decompose $F$ using proper $K_2$-cutsets as follows.

\begin{description}
\item[ Step 4.1:]
Let $\mathcal F'=\{ F \}$ and $\mathcal L'=\emptyset$.

\item[ Step 4.2:]
If $\mathcal F'=\emptyset$ then merge $\mathcal L'$ with $\mathcal L$
and go to Step 3.
Otherwise, remove a graph $H$ from $\mathcal F'$.

\item[ Step 4.3:]
Check whether $H$ has a $K_2$-cutset. If it does not, then add $H$ to 
$\mathcal L'$ and go to Step 4.2. Otherwise, let $S$ be a $K_2$-cutset of $H$.
Check whether $S$ is proper. If it is not, then output 
``$G\not\in \mathcal C_2$'' and stop. Otherwise, construct blocks of 
decomposition w.r.t.\ $S$, add them to $\mathcal F'$ and go to Step 4.2.
\end{description}
\end{description}

We first prove the correctness of this algorithm. Suppose the algorithm
terminates in Step 4.3 because it has identified a $K_2$-cutset of $H$
that is not proper. By Step~2, the graph $F$ that is placed in $\mathcal F'$
in Step 4.1 is 2-connected, and since blocks of decomposition of 
a 2-connected graph w.r.t.\ a $K_2$-cutset are also 2-connected, all graphs
that are ever placed in list $\mathcal F'$ are 2-connected, and in particular
$H$ is 2-connected. Therefore, by Lemma \ref{l:pK2cutset}, $H$ is correctly
identified as not belonging to $\mathcal C_2$.

We may now assume that the algorithm terminates in Step 3 by returning the
list $\mathcal L$. 
By Lemma~\ref{blocksC_1}  (i) and (ii) hold.
By the construction of the algorithm, clearly (iii) holds.

Let $\mathcal F^*$ be the list $\mathcal F$ at the end of Step 2.
Since every node of $G$ is in at most two graphs of $\mathcal F^*$ and
every edge of $G$ is in exactly one graph of $\mathcal F^*$, the following 
holds:
\\
\\
$(1) \hspace{2ex} \sum_{F\in \mathcal F^*} |V(F)|\leq 2n$ and 
  $\sum_{F\in \mathcal F^*} |E(F)|= m$.
\\
\\
Let $F$ be a graph placed in the list $\mathcal F'$ in Step 4.1, and let
$\mathcal L_F$ be the list $\mathcal L'$ at the time it is merged with
$\mathcal L$ in Step 4.2. We now show that:
\\
\\
$(2) \hspace{2ex} |\mathcal L_F| \leq |V(F)|$,  
  $\sum_{L\in \mathcal L_F} |V(L)|\leq 3|V(F)|$,
and $\sum_{L\in \mathcal L_F} |E(L)|\leq |V(F)|+|E(F)|$,
\\
\\
For any graph $T$ define $\phi (T)=|V(T)|-2$ and $\psi (T)=|E(T)|-1$. 
Suppose $(S,A,B)$ is a split
of a $K_2$-cutset of $H$ and let $H_A=H[S\cup A]$ and $H_B=H[S \cup B]$
be the blocks of decomposition. Clearly 
$\phi (H)=|A|+|B|=\phi (H_A)+\phi (H_B)$ and
$\psi (H)=\psi (H_A)+\psi (H_B)$. 
Since a block is of size at least 3,
it follows that $\phi (H),\phi (H_A)$, $\phi (H_B)$, 
$\psi (H),\psi (H_A)$, $\psi (H_B)$ 
are all at least 1.
Therefore $\phi (F)=\sum_{L\in \mathcal L_F} \phi (L) \geq |\mathcal L_F|$
and hence $|\mathcal L_F|\leq |V(F)|$. Furthermore, 
$|V(F)|-2=\sum_{L\in \mathcal L_F}(|V(L)|-2)$, and so
$\sum_{L\in \mathcal L_F}|V(L)|=|V(F)|-2+2|\mathcal L_F|\leq 3|V(F)|$.
By a similar argument, but using $\psi$,
$|E(F)|-1=\sum_{L\in \mathcal L_F}(|E(L)|-1)$, and so
$\sum_{L\in \mathcal L_F}|E(L)|=|E(F)|-1+|\mathcal L_F|\leq |V(F)|+|E(F)|$.
Therefore (2) holds.

Let $\mathcal L^*$ be the list $\mathcal L$ outputted by the algorithm. Then
\\
\\
\begin{eqnarray*}
\sum_{L\in \mathcal L^*}|V(L)| &=&
\sum_{F \in \mathcal F^*}\sum_{L\in \mathcal L_F} |V(L)|
\\
&\leq & \sum_{F \in \mathcal F^*}3|V(F)| 
\hspace{2ex} \mathrm{by} \hspace{1ex} (2)
\\
&\leq & 6n \hspace{2ex} \mathrm{by} \hspace{1ex} (1).
\end{eqnarray*}
\\
\\
Also
\\
\begin{eqnarray*}
\sum_{L\in \mathcal L^*}|E(L)| &=&
\sum_{F \in \mathcal F^*}\sum_{L\in \mathcal L_F} |E(L)|
\\
&\leq & \sum_{F \in \mathcal F^*}(|V(F)|+|E(F)| 
\hspace{2ex} \mathrm{by} \hspace{1ex} (2)
\\
&\leq &2n+m \hspace{2ex} \mathrm{by} \hspace{1ex} (1)
\end{eqnarray*}
\\
\\
and hence (iv) holds.

We now show that this algorithm can be implemented to run in $\mathcal O(nm)$
time. Step 2 can be implemented to run in $\mathcal O(n+m)$ time 
\cite{ht1, tarjan}.
Checking whether a graph $H$ has a $K_2$-cutset can be done by the algorithm
in \cite{ht2} in $\mathcal O(|V(H)|+|E(H)|)$ time. This algorithm finds 
triconnected components of $H$ in linear time, and in particular it finds all
$K_2$-cutsets of $H$ (and some $S_2$-cutsets). By (2) it follows that
Step 4 can be implemented to run in $\mathcal O(|V(F)||E(F)|)$ time.
Step 4 is applied to every graph $F \in \mathcal F^*$. Since
\\
\\
\begin{eqnarray*}
\sum_{F \in \mathcal F^*}|V(F)||E(F)| 
&\leq &\sum_{F\in \mathcal F^*}|V(F)|\sum_{F\in \mathcal F^*}|E(F)|
\\
&\leq & 2nm \hspace{2ex} \mathrm{by} \hspace{1ex} (1),
\end{eqnarray*}
\\
\\
it follows that the total running time is $\mathcal O(nm)$.
\end{proof}

A cycle of length $k$ is denoted by $C_k$.

\begin{lemma}\label{noT-S-C_5}
Let $G$ be a 2-connected graph that does not have a $K_2$-cutset.
  If $G\in \mathcal C_2$ and it contains a $C_k$, for some
  $k\in\{3,4,5\}$, as an induced subgraph, then  $G=C_k$.
\end{lemma}
\begin{proof}
  Let $G\in \mathcal C_2$ and suppose that $G$ contains a
  $C_k=x_1x_2\ldots x_kx_1$ as an induced subgraph, for some
  $k\in\{3,4,5\}$. Assume $G\neq C_k$ and that $G$ has no 1-cutset nor
  $K_2$-cutset. Let $K$ be a connected component of $G\setminus C_k$.

If a node $x\in K$ is adjacent to more than one node of $C_k$, then
$V(C_k)\cup \{x\}$ induces a propeller of $G$. So a node of $K$ can
have at most one neighbor in $C_k$. Since $G$ has no 1-cutset 
nor $K_2$-cutset, 
$|N(K)\cap V(C_k)|\geq 2$, and if $|N(K)\cap V(C_k)|=2$, then
the two nodes of $N(K)\cap V(C_k)$ are nonadjacent.

Suppose $k=3$, and let $P$ be a minimal path of $K$ such that its
endnodes are adjacent to different nodes of $C_k$. Then $V(P)\cup
V(C_k)$ induces a propeller. Therefore $k\in\{4,5\}$, and hence
$N(K)\cap V(C_k)$ contains nonadjacent nodes. Let $P$ be a minimal
path of $K$ such that its endnodes are 
adjacent to nonadjacent nodes of $C_k$. We may
assume w.l.o.g.\ that the endnodes of $P$ are adjacent to $x_1$ and
$x_3$. By the choice of $P$, we may assume w.l.o.g.\ that nodes of
$V(C_k)\setminus\{x_1,x_2,x_3\}$ have no neighbors in $P$. But then
$V(C_k)\cup V(P)$ induces a propeller.
\end{proof}

\begin{lemma}\label{l:S2preserve}
Let $G$ be a 2-connected graph with no $K_2$-cutset.
Let $(\{ u,v \} ,A,B)$ be a split of a proper $S_2$-cutset of $G$,
and $G_A$ and $G_B$ the corresponding
blocks of decomposition. Then the following hold.
\begin{itemize}
\item[(i)] $G_A$ and $G_B$ are 2-connected and have no $K_2$-cutset.
\item[(ii)] If $|A|\leq 2$ or $|B|\leq 2$, then $G \not\in \mathcal C_2$. 
\end{itemize}
\end{lemma}
\begin{proof}
Since  $G$ is connected, then clearly by the construction of blocks,
so are $G_A$ and $G_B$. To prove (i) assume w.l.o.g.\ that $S$ is a
1-cutset or a $K_2$-cutset of $G_A$. Since $G$ is 2-connected, both
$u$ and $v$ have a neighbor in every connected component of 
$G \setminus \{ u,v \}$. So we may assume that $S\cap \{ u',v'\}=\emptyset$
(where $u'$ and $v'$ are the marker nodes of $G_A$). Then w.l.o.g.\ we may
assume that $v\not\in S$. Let $C$ and $D$ be connected components of
$G_A\setminus S$ such that $v \in C$. Then $u',v'\in C$, and if $u\not\in S$
then $u\in C$. Therefore $D\subseteq A$, and hence $S$ is a cutset of $G$,
a contradiction. Therefore (i) holds.

To prove (ii) assume w.l.o.g.\ that $|A|\leq 2$ and $G \in {\cal
  C}_2$. Since $G$ is 2-connected, there is a chordless $uv$-path $P$
in $G[A \cup \{ u,v \}]$. Since $\{ u,v \}$ is a proper $S_2$-cutset,
it follows that $P$ has length 3, say $P=uxv$, and $|A|=2$, say $A=\{
x,y\}$. Since $G$ is 2-connected, $|N(y)\cap \{ u,v,x \}|\geq 2$, so
$G$ contains a cycle of length at most~4.  But then by Lemma
\ref{noT-S-C_5}, $G$ is a chordless cycle, a contradiction.
\end{proof}

\begin{theorem}\label{a:r2}
There is an algorithm with the following specifications.
\begin{description}
\item[ Input:] A 2-connected graph $G$ that does not have a $K_2$-cutset.
\item[ Output:]
YES if $G \in \mathcal C_1$, and NO otherwise.
\item[ Running time:] $\mathcal O(nm)$
\end{description}
\end{theorem}
\begin{proof}
Consider the following algorithm.
\begin{description}
\item[ Step 1:]
If $G$ has fewer than 7 nodes, then check directly whether $G\in \mathcal C_1$,
return the answer and stop.

\item[ Step 2:]
Let $\mathcal L  =\{ G \}$. 

\item[ Step 3:]
If $\mathcal L =\emptyset$ then output YES and stop.
Otherwise, remove a graph $F$ from $\mathcal L$.

\item[ Step 4:]
Check whether $F \in \mathcal C_0$. If it is then go to Step 3.
Otherwise, let $w$ be a node of $F$ and $u$ and $v$ its neighbors that
are of degree at least 3.

\item[ Step 5:]
If $uv$ is an edge then output NO and stop.
Otherwise, check whether $F\setminus w$ has a cutnode $w'$ that
separates $u$ from $v$.
If it does not then output NO and stop.
Otherwise, let $C_u$ and $C_v$ be the connected components of 
$F\setminus \{ w,w'\}$ that contain $u$ and $v$ respectively, and
denote by $C$ the union of the remaining components. If
$|C_u \cup C|\leq 2$ or $|C_v|\leq 2$ then output NO and stop.
Otherwise for the split $(\{ w,w'\} ,C\cup C_u,C_v)$ of the proper
$S_2$-cutset $\{ w,w'\}$ construct the corresponding blocks of decomposition,
add them to $\mathcal L$ and go to Step 3.
\end{description}

We first prove the correctness of this algorithm.
We may assume that the algorithm does not terminate in Step 1.
By Step 1 and Lemma \ref{l:S2preserve}, all the graphs that are ever
put on list $\mathcal L$ are 2-connected, have no $K_2$-cutset
and have at least 7 nodes. 
Note that by Lemma~\ref{blocksC_12}, at every stage of the algorithm
when blocks of decomposition are added to $\mathcal L$ in Step 5
the following holds: $G$ belongs to $\mathcal C_1$ if and only if
all the graphs in $\mathcal L$ belong to $\mathcal C_1$.
If the algorithm terminates in Step 5, then
by  Lemma \ref{l:S2preserve} and 
the proof of Theorem \ref{decompositionC_1} 
it does so correctly.
If the algorithm terminates in Step 3, then by Lemma~\ref{blocksC_12}
it does so correctly.

We now show that this algorithm can be implemented to run in
$\mathcal O(nm)$ time.
Step 1 can clearly be implemented to run in constant time.
For a graph $F$ that is removed from list $\mathcal L$ in Step 3, clearly
Steps 4 and 5 can be implemented to run in
$\mathcal O(|V(F)|+|E(F)|)$ time. If we show that the number of times
these steps are applied is at most $n$, then it follows that the total
running time is $\mathcal O(nm)$.

Let $\mathcal F$ be the set of graphs that are identified as belonging to
$\mathcal C_0$ in Step 4. Then the number of times Steps 4 and 5 are applied
is at most $|\mathcal F|$. We now show that $|\mathcal F|\leq n$.
For any graph $H$ define $\phi (H)=|V(H)|-6$. Now let $F$ be a graph
that is decomposed by a proper $S_2$-cutset in Step 5. Denote by
$(S,A,B)$ the split used for the decomposition, and by $F_A$ and $F_B$
the corresponding blocks of decomposition.
Clearly $\phi (F)=|A|+|B|-4=\phi (F_A)+\phi (F_B)$. Since all graphs that
are ever placed on list $\mathcal L$ have at least 7 nodes, it follows
that $\phi (F),\phi (F_A)$ and $\phi (F_B)$ are all at least 1.
Hence $\phi (G)=\sum_{L\in \mathcal F} \phi (L)\geq |\mathcal F|$.
Therefore $|\mathcal F|\leq n$.
\end{proof}

\begin{theorem}\label{a:r3}
There is an algorithm with the following specifications.
\begin{description}
\item[ Input:] A  graph $G$.
\item[ Output:]
YES if $G \in \mathcal C_1$, and NO otherwise.
\item[ Running time:] $\mathcal O(nm)$
\end{description}
\end{theorem}
\begin{proof}
First apply the algorithm from Theorem \ref{a:r1}. If this algorithm
returns $G\not\in \mathcal C_2$, then return NO and stop. Otherwise,
let $\mathcal L$ be the outputted list. Now apply the algorithm from
Theorem \ref{a:r2} to every graph in $\mathcal L$. If any of the
outputs is NO then return NO and stop, and otherwise return YES and stop.
Since $\sum_{L \in \mathcal L} |V(L)||E(L)|\leq 
\sum_{L \in \mathcal L} |V(L)|\sum_{L \in \mathcal L} |E(L)|$, it follows
by Theorems \ref{a:r1} and \ref{a:r2} that the running time is
$\mathcal O(nm)$.
\end{proof}

\subsection{Recognition algorithm for $\mathcal C_2$}
\label{ss:rac2}

We say that an $I$-cutset $\{ u,v,w\}$ is {\em proper} if no node of
$G \setminus \{ u,v,w \}$ has at least two neighbors in $\{ u,v,w\}$.
Let 
$(\{u,v,w\} , K', K'')$ be a split of a proper $I$-cutset of a graph $G$, 
and assume $uv$ is an edge.
The {\em blocks of decomposition} of $G$ w.r.t.\ this split are 
graphs $G'$ and $G''$ defined as follows.
Block $G'$ is the graph obtained from
$G[V(K')\cup\{u,v,w\}]$ by adding new nodes $u_1'$, $u_2'$, $v_1'$ 
and $v_2'$ (called the {\em marker nodes} of $G'$) 
and edges $uu_1'$, $u_1'u_2'$, $u_2'w$, $vv_1'$, $v_1'v_2'$ 
and
$v_2'w$. Block $G''$ is the graph obtained from
$G\left[V(K'')\cup\{u,v,w\}\right]$ by adding new nodes $u_1''$,
$u_2''$, $v_1''$ and $v_2''$ (called the {\em marker nodes} of $G''$) 
and edges $uu_1''$, $u_1''u_2''$,
$u_2''w$, $vv_1''$, $v_1''v_2''$ and $v_2''w$.

We use the following notation in the proofs that follow.  By the
definition of an $I$-cutset, $G[K' \cup \{u,v,w\}]\sm uv$ contains a
chordless $uv$-path $P'_{uv}$, a chordless $uw$-path $P'_{uw}$, and a
chordless $vw$-path $P'_{vw}$, whose interiors belong to the same
connected component of $G[K']$. Define $P''_{uv}$ , $P''_{uw}$ ,
$P''_{vw}$ for $K''$ in the obvious analogous manner.

\begin{lemma}\label{l:2nbrs}
  If $G$ is a 2-connected graph in $\mathcal C_2$ that has no
  $K_2$-cutset, then every  $I$-cutset of $G$ is proper. 
\end{lemma}

\begin{proof}
  Let $(\{ u,v,w \},K',K'')$ be a split of an $I$-cutset such that
  $uv$ is an edge.  Suppose that $x \in V(G) \setminus \{ u,v,w \}$
  has at least two neighbors in $\{ u,v,w \}$.  W.l.o.g.\ $x\in K'$
  and by Lemma \ref{noT-S-C_5} w.l.o.g.\ $N(x) \cap \{ u,v,w\} =\{ u,w
  \}$.  If $x \not \in V(P'_{uw})$ then $G[V(P'_{uw})\cup V(P''_{uw})
  \cup \{ x \} ]$ is a propeller, and hence $G\not\in \mathcal C_2$.
  So we may assume that $x \in V(P'_{uw})$, i.e.\ $P'_{uw}=uxw$.  Let
  $C'$ be the connected component of $G[K']$ that contains $x$, and
  let $P$ be a shortest $xv$-path in $G[C' \cup \{ v\} ]$. If $w$ does
  not have a neighbor in $P \setminus x$, then $G[V(P) \cup
  V(P''_{vw}) \cup \{ u \}]$ is a propeller with center $u$, and hence
  $G\not\in \mathcal C_2$.  So we may assume that $w$ has a neighbor
  in $P \setminus x$. If $u$ does not have a neighbor in $P \setminus
  x$, then $G[V(P)\cup \{ u,v,w \}]$ is a propeller, and hence
  $G\not\in \mathcal C_2$.  So, we may assume that $u$ has a neighbor
  in $P \setminus x$. Then $G[(V(P) \setminus \{ x,v\}) \cup \{ u,w
  \}]$ contains a chordless $uw$-path $Q$, and hence $G[V(Q) \cup
  V(P''_{uw}) \cup \{ x \} ]$ is a propeller, implying that $G\not\in
  \mathcal C_2$.
\end{proof}

\begin{lemma}\label{blocksC_2}
Let $G$ be a 2-connected graph that does not have a  $K_2$-cutset. 
Let 
$(\{u,v,w\} , K', K'')$ be a split of a proper $I$-cutset of $G$, and
$G'$ and $G''$ the corresponding 
blocks of decomposition.
Then $G\in\mathcal C_2$  if and only if 
$G'\in\mathcal C_2$\mbox{ and }$G''\in\mathcal C_2$.
\end{lemma}

\begin{proof}
Let $G\in \mathcal C_2$ and assume w.l.o.g.\ that $G'$ contains a 
propeller $(C,x)$ as an induced subgraph. 
Clearly $x\in K'\cup\{u,v,w\}$. 
Since $(C,x)$ cannot be contained in $G$,
$V(C) \cap \{ u_1',u_2',v_1',v_2'\} \neq \emptyset$. W.l.o.g.\ we may assume
that $uu_1'u_2'w$ is a subpath of $C$. If 
$V(C) \cap \{ v_1',v_2'\} \neq \emptyset$, then 
$V(C)=\{ u,v,u_1',u_2',v_1',v_2',w\}$, and since $x$ is adjacent to at least 
two nodes of $C$ it follows that $x$ is adjacent to at least two nodes of
$\{ u,v,w \}$, 
contradicting the assumption that $\{ u,v,w \}$ is a proper $I$-cutset.
Therefore $V(C) \cap \{ v_1',v_2' \} =\emptyset$. Let $P'$ be the 
$uw$-subpath of $C$ that does not contain $u_1'$. 
If $v \notin V(P')$ then $V(P') \cup V(P''_{uw}) \cup \{ x \}$
induces a propeller in $G$ with center $x$, a contradiction.
Hence $v \in V(P')$. 
If $x$ has at least two neighbors in $P'\setminus u$, then
$V(P') \cup V(P''_{vw}) \cup \{ x \}$ induces a propeller in $G$, 
a contradiction. So $N(x) \cap V(P')=\{ u,a \}$, where $a$ is a node 
of $P'\setminus \{ u,v,w \}$.
If $v$ does not have a neighbor in $P''_{uw}\setminus u$,
then $V(P''_{uw}) \cup V(P') \cup \{ x\}$ 
induces a propeller in $G$ with center
$x$, a contradiction. Hence $v$ has a neighbor in $P''_{uw}\setminus u$.
But then the $wa$-subpath of $P'$ together with $V(P''_{uw})\cup \{ x,v \}$
induces a propeller in $G$ with center $v$, a contradiction.

To prove the converse assume that $G'\in \mathcal C_2$ and
$G''\in\mathcal C_2$, but that $G$ contains as an induced subgraph
a propeller $(C,x)$.
Let us first assume that $C$ is contained in $G'$ or $G''$, w.l.o.g
$V(C)\subset V(G')$. If $x\in K'\cup\{u,v,w\}$ then $(C,x)$ is in
$G'$, a contradiction. 
Otherwise $x\in K''$, and hence it has at least two neighbors in
$\{u,v,w\}$, contradicting the assumption that $\{ u,v,w \}$ is a proper
$I$-cutset.
So $C$ must contain nodes from both $K'$ and $K''$, and
therefore it contains $w$ and at least one node from the set
$\{u,v\}$. W.l.o.g.\ we may assume that it contains $u$ and that $x\in
V(G')$. 
Let $P$ be
the $uw$-subpath of $C$
contained in $G'$. 
First let us assume that $x\neq v$.  
But 
then the node set $V(P)\cup \{x,u_1',u_2'\}$ induces 
a propeller in $G'$, a
contradiction. So $x=v$, and therefore $v$ is adjacent to a node $y$
of $C$ different from $u$. We may assume w.l.o.g.\ that $y\in G'$. 
But 
then the node set $V(P)\cup \{x,u_1',u_2'\}$ induces 
a propeller in $G'$, a contradiction.
\end{proof}

\begin{lemma}\label{l:Ipreserve}
Let $G$ be a 2-connected graph that does not have a $K_2$-cutset.
Let $(\{ u,v,w\},K',K'')$ be a split of a proper $I$-cutset of $G$, and
$G'$ and $G''$ the corresponding blocks of decomposition.
Then the following hold.
\begin{itemize}
\item[(i)] $G'$ and $G''$ are 2-connected and have no $K_2$-cutset.
\item[(ii)] If $|K'|\leq 4$ or $|K''|\leq 4$ then $G\not\in \mathcal C_2$.
\end{itemize}
\end{lemma}
\begin{proof}
W.l.o.g.\ $uv$ is an edge.
Since  $G$ is connected, then clearly by the construction of blocks,
so are $G'$ and $G''$. To prove (i) assume w.l.o.g.\ that $S$ is a
1-cutset or a $K_2$-cutset of $G'$. Since $G$ is 2-connected
and has no $K_2$-cutset, 
every connected component of 
$G \setminus \{ u,v,w \}$ must contain a neighbor of $w$ and a
neighbor of $u$ or $v$. 
So we may assume that $S$
does not contain any of  the marker nodes of $G'$. Then w.l.o.g.\ we may
assume that $v\not\in S$. Let $C$ and $D$ be connected components of
$G'\setminus S$ such that $v \in C$. Then all the marker nodes and
nodes of $\{ u,w\} \setminus S$ are in
$C$. Therefore $D\subseteq K'$, and hence $S$ is a cutset of $G$,
a contradiction. Therefore (i) holds.

To prove (ii) assume w.l.o.g.\ that $|K'|\leq 4$. 
Then $P'_{uw}$ is of length at most 5. If $v$ has a neighbor on
$P'_{uw}\setminus u$, then since $\{ u,v,w \}$ is a proper $I$-cutset
and by Lemma \ref{noT-S-C_5},
$G \not\in \mathcal C_2$.
So we may assume that  $v$ does not have a neighbor on $P'_{uw}\setminus u$.
Since $\{ u,v,w \}$ is a proper $I$-cutset, $P'_{uw}$ and $P'_{vw}$ are
both of length at least 3. 
Suppose that
the interior nodes of $P'_{uw}$ and $P'_{vw}$ are disjoint.
Then $P'_{uw}=ux_1x_2w$ and 
$P'_{vw}=vy_1y_2w$, and hence since $x_1,x_2, y_1,y_2$ all belong to the same 
connected component of $G\setminus \{ u,v,w \}$, there must be an edge
between a node of $\{ x_1,x_2\}$ and a node of $\{ y_1,y_2\}$. But then
by Lemma \ref{noT-S-C_5}, $G \not\in \mathcal C_2$.
Finally we may assume w.l.o.g.\ that $P'_{uw}=ux_1x_2x_3w$ and
$P'_{vw}=uy_1x_3w$ (else by   Lemma \ref{noT-S-C_5}, $G \not\in \mathcal C_2$).
But then either $G[K' \cup \{ u,v,w\} \cup V(P''_{vw})]$ is a
propeller with center $y_1$ (if $u$ does not have a neighbor on
$P''_{vw} \setminus v$) or
 $G[ \{ u,v,w\} \cup V(P'_{vw}) \cup V(P''_{vw})]$ is a
propeller with center $u$ (if $u$ does have a neighbor on
$P''_{vw} \setminus v$).
\end{proof}

\begin{theorem}\label{a:r4}
There is an algorithm with the following specifications.
\begin{description}
\item[ Input:] A 2-connected graph $G$ that does not have a $K_2$-cutset.
\item[ Output:] 
YES if $G \in \mathcal C_2$, and NO otherwise.
\item[ Running time:] $\mathcal O(n^2m^2)$
\end{description}
\end{theorem}

\begin{proof}
Consider the following algorithm.
\begin{description}
\item[ Step 1:]
If $G$ has fewer than 12 nodes, 
then check directly whether $G\in \mathcal C_2$,
return the answer and stop.

\item[ Step 2:]
Let $\mathcal L  =\{ G\}$. 

\item[ Step 3:]
If $\mathcal L =\emptyset$ then output YES and stop.
Otherwise, remove a graph $F$ from $\mathcal L$.

\item[ Step 4:]
Use the algorithm from Theorem \ref{a:r2} to
check whether $F \in \mathcal C_1$. If it is then go to Step 3.

\item[ Step 5:]
For every edge $uv$ and node $w$ of $F$, check whether $\{ u,v,w\}$ is
a proper $I$-cutset of $F$. If such a cutset does not exist,
return NO and stop.
Otherwise let $(\{ u,v,w \},K',K'')$ be a split of a proper $I$-cutset
of $G$. If $|K'|\leq 4$ or $|K''|\leq 4$, then
 return NO and stop.
Otherwise construct
blocks of decomposition, add them to $\mathcal L$ and go to Step 3.

\end{description}

We first prove the correctness of this algorithm.
We may assume that the algorithm does not terminate in Step 1.
By Step 1 and Lemma \ref{l:Ipreserve}, all the graphs that are ever
put on list $\mathcal L$ are 2-connected and have no $K_2$-cutset.
So the correctness of the algorithm follows from 
Theorem \ref{decompositionC_2}, Lemma \ref{l:2nbrs}, Lemma \ref{blocksC_2}
and Lemma \ref{l:Ipreserve}.

We now show that this algorithm can be implemented to run in
$\mathcal O(n^2m^2)$ time.
Step 1 can clearly be implemented to run in constant time.
For a graph $F$ that is removed from list $\mathcal L$ in Step 3,
Step 4 runs in
$\mathcal O(|V(F)||E(F)|)$ time, and 
Step 5 can be implemented to run in $\mathcal O(|V(F)||E(F)|^2)$ time
(since there are $\mathcal O(|V(F)||E(F)|)$ sets $\{ u,v,w\}$
that need to be checked, and for each one of them checking whether
it is a proper $I$-cutset can clearly be done in
$\mathcal O(|V(F)|+|E(F)|)$ time).
If we show that the number of times
these steps are applied is at most $n$, then it follows that the total
running time is $\mathcal O(n^2m^2)$.

Let $\mathcal F$ be the set of graphs that are identified as belonging to
$\mathcal C_1$ in Step 4. Then the number of times Steps 4 and 5 are applied
is at most $|\mathcal F|$. We now show that $|\mathcal F|\leq n$.
For any graph $H$ define $\phi (H)=|V(H)|-11$. Now let $F$ be a graph
that is decomposed by a proper $I$-cutset in Step 5. Denote by
$(S,K',K'')$ the split used for the decomposition, and by $F'$ and $F''$
the corresponding blocks of decomposition.
Clearly $\phi (F)=|A|+|B|-8=\phi (F')+\phi (F'')$. Since 
by Step 1 and Lemma \ref{l:Ipreserve},
all graphs that
are ever placed on list $\mathcal L$ have at least 12 nodes, it follows
that $\phi (F),\phi (F')$ and $\phi (F'')$ are all at least 1.
Hence $\phi (G)=\sum_{L\in \mathcal F} \phi (L)\geq |\mathcal F|$.
Therefore $|\mathcal F|\leq n$.
\end{proof}

\begin{theorem}\label{a:r5}
There is an algorithm with the following specifications.
\begin{description}
\item[ Input:] A graph $G$.
\item[ Output:] 
YES if $G \in \mathcal C_2$, and NO otherwise.
\item[ Running time:] $\mathcal O(n^2m^2)$
\end{description}
\end{theorem}
\begin{proof}
First apply the algorithm from Theorem \ref{a:r1}. If this algorithm
returns $G\not\in \mathcal C_2$, then return NO and stop. Otherwise,
let $\mathcal L$ be the outputted list. Now apply the algorithm from
Theorem \ref{a:r4} to every graph in $\mathcal L$. If any of the
outputs is NO then return NO and stop, and otherwise return YES and stop.
Since $\sum_{L \in \mathcal L} |V(L)|^2|E(L)|^2 \leq 
(\sum_{L \in \mathcal L} |V(L)|)^2(\sum_{L \in \mathcal L} |E(L)|)^2$, 
it follows
by Theorems \ref{a:r1} and \ref{a:r4} that the running time is
$\mathcal O(n^2m^2)$.
\end{proof}

\section{Flat edges and edge-coloring}
\label{sec:extr}

A {\em flat edge} of a graph $G$ is an edge both of whose endnodes are
of degree 2. In this section we show that every 2-connected
propeller-free graph has a flat edge and use this property to
edge-color it. To do this we first show the existence of an {\em
  extreme decomposition}, i.e.\ a decomposition in which one of the
blocks is in ${\cal C}_0$.

\begin{lemma}\label{fe1}
  Let $G$ be a 2-connected graph in ${\cal C}_2 \sm {\cal C}_0$. Then,
  there exists $S \subseteq V (G)$ such that (i) $S$ is either a
  proper $I$-cutset or a proper $S_2$-cutset of $G$, (ii) there exists
  a split $(S,K',K'')$ such that at least one of the blocks of
  decomposition, say $G'$, is in ${\cal C}_0$, and (iii) all nodes in
  $S$ are of degree at least three in $G'$.
\end{lemma}

\begin{proof}
  By Lemma \ref{l:pK2cutset}, Theorem \ref{decompositionC_1} and
  Theorem \ref{decompositionC_2}, $G$ has an $I$-cutset or a proper
  $S_2$-cutset.  Note that by Lemma~\ref{l:2nbrs}, any $I$-cutset is
  proper.  Let $(S,K',K'')$ be a split of an $I$-cutset or a proper
  $S_2$-cutset of $G$ such that among all such splits, $|K'|$ is
  minimized.  Let $G'$ be the block of decomposition that contains
  $K'$.  If $S$ is a proper $S_2$-cutset we let $S=\{ u,v\}$, and if
  $S$ is an $I$-cutset we let $S=\{ u,v,w \}$ and assume that $uv$ is
  an edge.

\vspace{2ex}

\noindent
{\bf Claim 1:} {\em $G'$ is 2-connected, has no $K_2$-cutset and belongs
to ${\cal C}_2$.}

\vspace{2ex}

\noindent
{\em Proof of Claim 1:}
$G'$ is 2-connected and has no $K_2$-cutset by Lemma \ref{l:S2preserve}
and Lemma \ref{l:Ipreserve}.
If $S$ is an $I$-cutset, then $G'\in {\cal C}_2$ by Lemma \ref{blocksC_2}.
So suppose that $S$ is a proper $S_2$-cutset. Since $G$ is 2-connected,
$G[S\cup K'']$ contains a $uv$-path $P$. If $G'$ contains a propeller
as an induced subgraph, then so does $G[S\cup K' \cup V(P)]$.
Therefore, $G \in {\cal C}_2$.
This completes the proof of Claim 1.

\vspace{2ex}

\noindent
{\bf Claim 2:} {\em If $S$ is a proper $S_2$-cutset, then both $u$ and $v$
have at least two neighbors in $K'$.
In particular, all nodes of $S$ have degree at least 3 in $G'$.}

\vspace{2ex}

\noindent
{\em Proof of Claim 2:}
Suppose not and let $u_1$ be the unique neighbor of $u$ in $K'$.  By
Claim 1 and Lemma \ref{noT-S-C_5} (applied to $G'$), $u_1v$ is not an
edge. But then $(\{ u_1,v\} ,K'\setminus \{ u_1\},K'' \cup \{ u \})$
is a split of a proper $S_2$-cutset of $G$, contradicting our choice
of $(S,K',K'')$.  This completes the proof of Claim~2.

\vspace{2ex}

We now show that $G' \in {\cal C}_0$. Assume not. By Claim 1,
Lemma \ref{l:pK2cutset}, Theorem \ref{decompositionC_1} and
Theorem \ref{decompositionC_2}, $G'$ has an $I$-cutset or a 
proper $S_2$-cutset with split $(C,C_1,C_2)$.
W.l.o.g.\ we may assume that 
$(C,C_1,C_2)$ is chosen so that $|C_i|$, for some $i\in\{ 1,2\}$, is minimized.
Let $M$ be the set of marker nodes of $G'$. By Claims 1 and 2 (applied to
$G'$ and $C$),
all nodes of $C$ have degree at least 3 in $G'$, and hence 
$C\cap M=\emptyset$.
We now consider the following two cases.

\vspace{2ex}

\noindent{\bf Case 1:} $S$ is a proper $S_2$-cutset of $G$.
\\
W.l.o.g.\ $M \subseteq C_2$. Note that $C_1$ is a proper subset of $K'$.
But then $(C,C_1,(C_2\setminus M)\cup K'')$ is a split of an $I$-cutset
or a proper $S_2$-cutset of $G$, contradicting our choice of $(S,K',K'')$.

\vspace{2ex}

\noindent{\bf Case 2:} $S$ is an $I$-cutset of $G$.
\\
By the choice of $(S,K',K'')$, $G[K']$ is connected.  In particular,
$|C\cap S|\leq 2$.  If $|C\cap S|\leq 1$, then w.l.o.g.\ $(S\cup
M)\setminus C\subseteq C_2$, and hence $(C,C_1,(C_2\setminus M)\cup
K'')$ is a split of an $I$-cutset or a proper $S_2$-cutset of $G$,
contradicting our choice of $(S,K',K'')$.  So $|C\cap S|=2$.  Since
each node of $S$ has a neighbor in $K'$, it follows that $C$ is an
$I$-cutset.  Suppose that marker nodes $u_1',u_2'$ are in $C_1$ and
$v_1',v_2'$ are in $C_2$. Then, w.l.o.g.\ $\{u, w\} \subseteq C$, so
$(C,C_1\setminus \{ u_1',u_2'\}, C_2\cup \{ u_1',u_2'\})$ is also a
split of an $I$-cutset of $G'$. So we may assume that w.l.o.g.  $(S
\cup M)\setminus C\subseteq C_2$, and hence $(C,C_1,(C_2\setminus
M)\cup K'')$ is a split of an $I$-cutset of $G$, contradicting our
choice of $(S,K',K'')$.
\end{proof}

A \emph{flat pair} in a graph $G$ is a pair of distinct flat edges $e,
f$ such that $e = uv$, $f=xy$, and $G[\{u, v, x, y\}]$ has exactly two
edges: $e$ and $f$.

\begin{lemma}
  \label{l:flatP}
  Let $G \in {\cal C}_0$ be a 2-connected graph. If $x\in V(G)$ is a
  node of degree at least~3, then there exists a flat pair $e$, $f$ of
  $G$ such that $e$ and $f$ both contain a node adjacent to $x$.
\end{lemma}

\begin{proof}
  Since $G$ is 2-connected, all nodes have degree at least two.  Since
  $G \in {\cal C}_0$, $x$ has at least two neighbors $u$ and $v$ of
  degree~2.  Since $G$ is 2-connected, $uv\notin E(G)$ (otherwise $x$
  is a cutnode). Let $u'$ (resp.\ $v'$) be the neighbor of $u$
  (resp.\ $v$) that is distinct from $x$. Since $G\in {\cal C}_0$ and
  $x$ has degree at least 3, both $u'$ and $v'$ are of degree 2.
  Since $G$ is 2-connected, $u'\neq v'$ and $u'v' \notin E(G)$
  (otherwise $x$ is a cutnode).  It follows that $uu', vv'$ is a flat
  pair.
\end{proof}

\begin{theorem}\label{fe2}
  If $G \in {\cal C}_2$ is 2-connected, then either $G$ is a
  chordless cycle or $G$ has a flat pair.
\end{theorem}

\begin{proof}
  We prove the result by induction on $|V(G)|$.  It is true when
  $|V(G)|\leq 3$.

  \vspace{2ex}

  \noindent {\bf Case 1:} {$G\in {\cal C}_0$}.  

  Follows directly
  from Lemma~\ref{l:flatP}. This completes the proof in Case~1.

\vspace{2ex}

\noindent {\bf Case 2:} {$G$ has a $K_2$-cutset}.  

Suppose $(\{
a,b \}, C_1,C_2)$ is a split of a $K_2$-cutset of $G$, and let $G_1$
and $G_2$ be the corresponding blocks of decomposition.  Note that by
Lemma \ref{l:pK2cutset}, $\{ a,b\}$ is a proper $K_2$-cutset.  For
$i=1,2$, $G_i$ is clearly 2-connected, by Lemma~\ref{blocksC_1}
$G_i\in {\cal C}_2$, and hence, by the induction hypothesis, $G_i$ is
either an induced cycle or it has a flat pair. Since $\{ a,b\}$ is
proper, $G_i$ cannot be a triangle.  Therefore, $G_i$ has a flat edge
entirely contained in $C_i$. Hence, $G$ has a flat pair formed by a
flat edge in $C_1$ and a flat edge in $C_2$.  This completes the
proof in Case~2.

\vspace{2ex}

From here on, we assume that $G$ has no $K_2$-cutset.  By
Lemma~\ref{fe1}, we may now assume that $G$ has an $I$-cutset or a
proper $S_2$-cutset $S$.  Moreover, there exists a split $(S,K',K'')$
such that the block of decomposition $G'$ that contains $K'$ belongs
to ${\cal C}_0$, and all nodes of $S$ have degree at least three in
$G'$.  This leads us to the following two cases.

\vspace{2ex}

\noindent {\bf Case 3:} {$S$ is an $I$-cutset}.  

Suppose $S = \{ u,v,w \}$ and $uv \in E(G)$.  Note that $w$ has degree
at least~3 in $G'$.  Since $G' \in {\cal C}_0$, $u$ and $v$ both have
a neighbor in $K'$, respectively $u'$ and $v'$, of degree~2.  Since
$G$ has no $K_2$-cutset, $u'\neq v'$ and $u'v'\notin E(G)$ (otherwise
$\{u, v\}$ is a $K_2$-cutset).  Since $G \in {\cal C}_0$ and $u, v$ have
degree at least~3, $u'$ (resp.\ $v'$) has one neighbors $u''$
(resp. $v''$) of degree~2.  Since $G$ has no $K_2$-cutset, $u''\neq
v''$ and $u''v''\notin E(G)$.  It follows that $u'u'', v'v''$ is a
flat pair in $G$.  This completes the proof in Case~3.

\vspace{2ex}

\noindent {\bf Case 4:} {$S$ is an $S_2$-cutset}. 

 Suppose $S =
\{u, v\}$.  By Lemma~\ref{l:S2preserve}, $G'$ is 2-connected.
If $u$ and $v$ are the only nodes of degree at least~3 in $G'$, then
$G$ is formed by at least three $uv$-paths.  Since $G' \in {\cal
  C}_0$, all these paths have length at least~3.  Therefore, they all
have an internal flat edge, and $G'$ has a flat pair entirely
contained in $K'$, that is therefore also a flat pair of $G$.
Otherwise, there is a node $x\in K'$ of degree at least~3.  By
Lemma~\ref{l:flatP}, $G'$ has a flat pair entirely contained in $K'$,
that is therefore also a flat pair of $G$. This completes the proof in
Case~4.
\end{proof}

An edge of a graph is \emph{pending} if it contains at least one node
of degree~1.

\begin{corollary}\label{edgeofdegree2}
  Every graph $G$ in $\mathcal C_2$ with at least one edge contains 
  an edge that is pending or flat.  
\end{corollary}

\begin{proof}
  We consider the classical decomposition of $G$ into blocks, in the
  sense of 2-connectivity (see~\cite{bondy.murty:book}). So, $G$ has a
  block $B$ that is either a pending edge of $G$, or a 2-connected
  graph containing at most one vertex $x$ that has neighbors in $V(G)
  \sm V(B)$.  In the latter case, by Theorem~\ref{fe2}, $B$ is either
  a chordless cycle or it has a flat pair, and so at least one flat
  edge of $B$ is non-incident to $x$, and is therefore a flat edge of
  $G$.
\end{proof}

An edge-coloring  of $G$ is a function $\pi:E\rightarrow C$ 
such that no two adjacent edges receive the same color $c\in C$. 
If $C=\{1,2,\ldots,k\}$, we say that $\pi$ is a $k$-edge coloring. 
The chromatic index of $G$, denoted by $\chi'(G)$, 
is the least $k$ for which $G$ has a $k$-edge-coloring.

Vizing's theorem states that $\chi'(G)=\Delta(G)$ or
$\chi'(G)=\Delta(G)+1$, where $\Delta(G)$ is maximum degree of nodes
in $G$. The edge-coloring problem or chromatic index problem is the
problem of determining the chromatic index of a graph. The problem is
NP-hard for several classes of graphs, and its complexity is unknown
for several others.  In this section we solve the edge-coloring
problem for the class $\mathcal C_2$.

\begin{theorem}
If $G$ is a graph in $\mathcal C_2$ such that $\Delta(G)\geq 3$, 
then $\chi'(G)=\Delta(G)$.
\end{theorem}

\begin{proof}
  Induction on $|E(G)|$.  If $|E(G)| = 0$, the result clearly holds.
  By Corollary~\ref{edgeofdegree2}, $G$ has an edge $ab$ that is
  pending or flat.  Note that ${\cal C}_2$ is not closed under
  removing edges in general, but it is closed under removing flat or
  pending edges.  Set $G' = (V(G), E(G) \sm \{ab\})$.  If $\Delta(G')
  \geq 3$, then by the induction hypothesis, we can edge-color $G'$
  with $\Delta(G')$ colors.  Otherwise, $\Delta(G') \leq 2$, so $G'$ is
  3-edge colorable.  In either cases, $G'$ is $\Delta(G)$-colorable.
  We can extend the edge-coloring of $G'$ to an edge-coloring of $G$
  as follows.  When $ab$ is pending, by assigning a color to $ab$ not
  used among the edges incident to $ab$, and when $ab$ is flat by
  assigning to $ab$ a color not used for the two edges adjacent to
  $ab$.
\end{proof}

Note that when $\Delta(G) \leq 2$, $G$ is a disjoint union of cycles
and paths, so $\chi'$ is easy to compute.  The proof above is easy to
transform into a polynomial time algorithm that outputs the coloring
whose existence is proved.

\end{document}